\documentclass[12pt,onecolumn]{IEEEtran}
\usepackage{epsf}
\usepackage{graphicx}
\usepackage{amsmath} \usepackage{amssymb}

\newtheorem{theorem}{Theorem}[section]

\newtheorem{lemma}{Lemma}[section]

\long\def\symbolfootnote[#1]#2{\begingroup%
\def\thefootnote{\fnsymbol{footnote}}\footnote[#1]{#2}\endgroup}
  
 \def\LL{{\cal L}} 
  
\def\BB{{\cal B}} \def\DD{{\cal D}}  
 \def\II{{\cal I}} \def\MM{{\cal M}}

\def\dref#1{(\ref{#1})}

\def\argmin{\mathop{\mbox{argmin}}\limits}

\def\be{\begin{equation}} \def\ee{\end{equation}}
\def\ba{\begin{array}} \def\ea{\end{array}} \def\bna{\begin{eqnarray}}
\def\ena{\end{eqnarray}}

 \def\NN{{\cal N}} \def\MM{{\cal M}}

\def\CC{{\cal C}}
\def\DD{{\cal D}}
\def\SS{{\cal S}}
 \def\XX{{\cal X}}   
\def\DD{{\cal D}}  
\def\TT{{\cal T}}
\def\II{{\cal I}}
\def\YY{{\cal Y}}

 \def\bna{\begin{eqnarray}}
\def\ena{\end{eqnarray}} \def\dref#1{(\ref{#1})}
\begin{document}
\title{On the Optimal Compressions in the Compress-and-Forward Relay Schemes}

\author{\authorblockN{Xiugang Wu and Liang-Liang Xie}\\
\authorblockA{
University of Waterloo, Waterloo, ON, Canada N2L 3G1 \\
Email: x23wu@uwaterloo.ca, llxie@uwaterloo.ca} }

\maketitle

\begin{abstract}
In the classical compress-and-forward relay scheme developed by (Cover and El Gamal, 1979), the decoding process operates in a successive way: the destination first decodes the compression of the relay's observation, and then decodes the original message of the source. Recently, several modified compress-and-forward relay schemes were proposed, where, the destination jointly decodes the compression and the message, instead of successively. Such a modification on the decoding process was motivated by realizing that it is generally easier to decode the compression jointly with the original message, and more importantly, the original message can be decoded even without completely decoding the compression. Thus, joint decoding provides more freedom in choosing the compression at the relay.

However, the question remains whether this freedom of selecting the compression necessarily improves the achievable rate of the original message. It has been shown in (El Gamal and Kim, 2010) that the answer is negative in the single-relay case. In this paper, it is further demonstrated that in the case of multiple relays, there is no improvement on the achievable rate by joint decoding either. More interestingly, it is discovered that any compressions not supporting successive decoding will actually lead to strictly lower achievable rates for the original message. Therefore, to maximize the achievable rate for the original message, the compressions should always be chosen to support successive decoding. Furthermore, it is shown that any compressions not completely decodable even with joint decoding will not provide any contribution to the decoding of the original message.

The above phenomenon is also shown to exist under the repetitive encoding framework recently proposed by (Lim, Kim, El Gamal, and Chung, 2010), which improved the achievable rate in the case of multiple relays. Here, another interesting discovery is that the improvement is not a result of repetitive encoding, but the benefit of delayed decoding after all the blocks have been finished. The same rate is shown to be achievable with the simpler classical encoding process of (Cover and El Gamal, 1979) with a block-by-block backward decoding process.

\end{abstract}

\section{Introduction}

The relay channel, originally proposed in \cite{van71}, models a
communication scenario where there is a relay node that can help the
information transmission between the source and the destination. Two
fundamentally different relay strategies have been
developed in \cite{covelg79} for such channels, which, depending on
whether the relay decodes the information or not, are generally known
as {\it decode-and-forward} and {\it compress-and-forward}
respectively. The compress-and-forward relay strategy is used when the
relay cannot decode the message sent by the source, but still can help
by compressing and forwarding its observation to the destination.
Specifically, consider the relay channel depicted in Fig. \ref{fig1}.
The relay compresses its observation $Y_1$ into $\hat Y_1$, and then
forwards $\hat Y_1$ to the destination via $X_1$. To reduce the rate
loss caused by the delay, block Markov coding was used in
\cite{covelg79}, and more blocks leads to less loss.
\begin{figure}[hbt]
\centering
\includegraphics[width=2.6in]{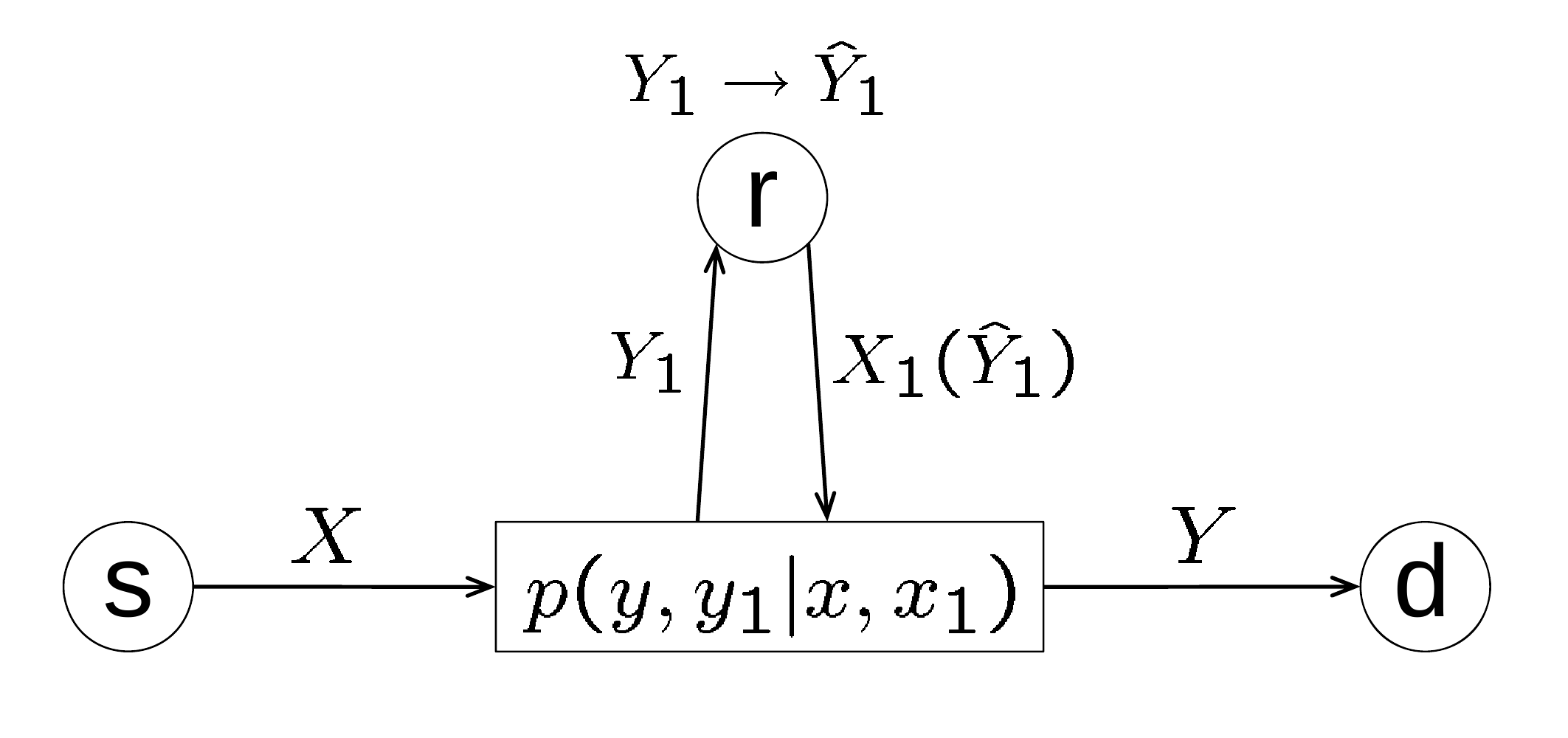}
\caption{The single relay channel.}
  \label{fig1}
\end{figure}

In this paper, based on the differences in the detailed
encoding/decoding processes, the following five different
compress-and-forward relay schemes will be considered.
\begin{itemize}
   \item Cumulative encoding/block-by-block forward decoding/compression-message successive decoding;
   \item Cumulative encoding/block-by-block forward decoding/compression-message joint decoding;
   \item Repetitive encoding/all blocks united decoding/compression-message joint decoding;
   \item Cumulative encoding/block-by-block backward decoding/compression-message successive decoding;
   \item Cumulative encoding/block-by-block backward decoding/compression-message joint decoding.
\end{itemize}

The Cumulative encoding/block-by-block forward decoding/compression-message successive decoding refers to
the original compress-and-forward scheme developed in \cite{covelg79}. The encoding is ``cumulative'' in the sense that in each new block, a new piece of information is
encoded at the source. This distinguishes from a ``repetitive''
encoding process recently proposed in \cite{KimElGamal}, where the
same information is encoded in each block. The decoding is named ``block-by-block forward'' to distinguish from the other two choices, where the decoding starts only after all the blocks have been finished, either by decoding with all the blocks together, or by decoding block-by-block backwardly. The decoding is also called ``compression-message successive'' in the sense that the destination first decodes the compression of the relay's observation, and then decodes the original message. The compression
$\hat Y_1$ can be first recovered at the
destination, as long as the following constraint is satisfied:
\begin{equation}
\label{constraint}
I(X_1;Y) > I(Y_1;\hat Y_1|X_1,Y).
\end{equation}
Then, based on $\hat Y_1$ and $Y$, the destination can decode the
original message $X$ if the rate of the original message satisfies
\begin{equation}
\label{cov}
R < I(X;\hat Y_1,Y|X_1).
\end{equation}

The above two-step compression-message successive decoding process requires $\hat Y_1$ to be decoded first. This facilitates the decoding of $X$, but is not a requirement of the original problem. Recognizing this, a joint compression-message decoding process was proposed in \cite{xie09}, where, instead of successively, the destination decodes $\hat Y_1$ and $X$ together. It turns out that the decoding of $X$ can be helped even if $\hat Y_1$ cannot be decoded first. In fact, with joint decoding, the constraint \dref{constraint} is not necessary, and instead of \dref{cov}, the achievable rate is expressed as
\begin{equation}
\label{rate1}
R<I(X;\hat Y_1,Y|X_1)-\max\{0,I(Y_1;\hat Y_1|X_1,Y)-I(X_1;Y)\}.
\end{equation}
Moreover, although $\hat Y_1$ is not even required to be decoded eventually, it can be more easily decoded by joint decoding, and instead of \dref{constraint}, we need a less strict constraint:
\begin{equation}
\label{constraint2}
I(X_1;Y) > I(Y_1;\hat Y_1|X_1,Y,X),
\end{equation}
where, it is clear to see the assistance provided by $X$.

Similar formulas as \dref{rate1} have been derived with different arguments in \cite{ElGamal}-\cite{ElGamalKim}.\footnote{The formula and proof in \cite{ElGamal} missed a $Y$, and were later corrected in \cite{ElGamalKim}.}

Therefore, compared to successive decoding, joint compression-message decoding provides more freedom in choosing the compression $\hat Y_1$. However, the question remains whether joint decoding achieves strictly higher rates for the original message than successive decoding. For the single relay case, it has been proved in \cite{ElGamalKim} that the answer is negative, and any rate achievable by either of them can always be achieved by the other. In this paper, we are going to further consider the case of multiple relays as depicted in Fig. \ref{fig2}, and demonstrate that joint decoding won't be able to achieve any higher rates either.  More interestingly, we will show that any compressions not supporting successive decoding will actually result in strictly lower achievable rates for the original message. Therefore, to optimize the achievable rate, the compressions should always be chosen so that successive decoding can be carried out.

\begin{figure}[hbt]
\centering
\includegraphics[width=2.8in]{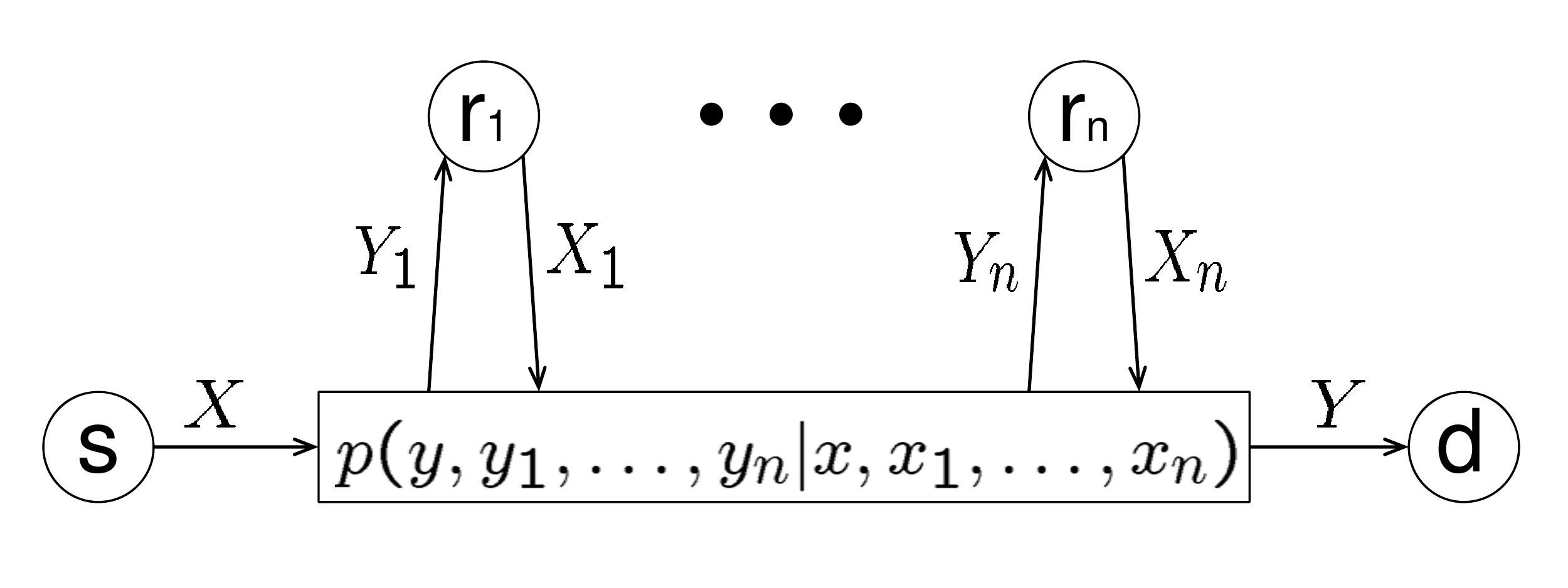}
\caption{The multiple-relay channel.}
 \label{fig2}
\end{figure}

Recently, a different encoding process was proposed in \cite{KimElGamal}, where instead of piece by piece, all the information is encoded in each block, and different blocks use independent codebooks to transmit the same information. Compared to cumulative encoding, this repetitive encoding process appears to introduce collaboration among all the blocks, so that all the blocks can unitedly contribute to the decoding of the same message. This repetitive encoding/all blocks united decoding process was combined with joint compression-message decoding in \cite{KimElGamal}, and although no improvement was shown in the single relay case, some interesting improvement on the achievable rate was obtained in the case of multiple relays. In this paper, we will show that actually it is not necessary to use repetitive encoding to introduce such collaboration among the blocks. The same rate can be achieved with cumulative encoding as long as the decoding starts after all the blocks have been finished. We will show that either by all blocks united decoding, or by block-by-block backward decoding, the same achievable rate can be obtained. Therefore, in terms of complexity, cumulative encoding/block-by-block backward decoding provides the simplest way to achieve the highest rate in the case of multiple relays.

Similarly, for these new encoding/decoding schemes, we will also show that the optimal compressions must be able to support successive compression-message decoding, and any compressions not supporting successive decoding will necessarily lead to strictly lower achievable rates than the optimal. Therefore, for any of these compress-and-forward relay schemes mentioned above, we can restrict our attention to successive compression-message decoding in the search for the optimal compressions of the relays' observations.\footnote{Part of the results were presented in \cite{WuXieAllerton}.} Of course, it should be noted that any compressions supporting successive decoding also support joint decoding.

Although the compressions supporting successive decoding can be explicitly characterized as we will show later, it is also of interest to consider other compressions not supporting successive decoding. For example, in a network with multiple destinations, when a relay is simultaneously helping more than one destinations, it is very likely that different destinations require different optimal compressions from the relay. In such a situation, the relay may have to find a tradeoff between these requirements, i.e., adopting a compression which may be too coarse for some destinations, but too fine, thus not supporting successive decoding, for the others. An example of this tradeoff to optimize the sum rate was given for the two-way relay channel in \cite{KimElGamal}. Another possibility of using too coarse or too fine compressions is when there is channel uncertainty, e.g., in wireless fading channels, so that it is impossible to accurately determine the optimal compressions even with explicit formulas. Therefore, it is of interest to study how coarser or finer compressions than the optimal affect the achievable rate of the original message \cite{WuXieVTC}.

It is not surprising that coarser compressions than the optimal do not fully exploit the capability of the relay, thus leading to lower achievable rates for the original message. However, it may not be so obvious why finer compressions will also lead to lower achievable rates. For this, one needs to realize that a relay's observation not only carries information about the original message, but also reflects the dynamics of the source-relay link, which is unrelated to the original message. Thus, compared to the direct link between the source and the destination, the support by the relay-destination link is not so pure. When the compression is too fine so that only joint compression-message decoding can be carried out, i.e., the direct source-destination link has to sacrifice, the gain does not make up for the loss. Furthermore, to the extreme, when the compression cannot be decoded even with joint decoding, the relay-destination link becomes useless, and the destination would rather simply treat the relay's input as purely noise in the decoding, as we will demonstrate in the paper.

The remainder of the paper is organized as the following. In Section \ref{S:mainresults}, we formally state our problem setup and summarize the main results. Then, in Section \ref{S:B-Bforward} and Section \ref{S:DecodingAllBlocks}, detailed proofs of the achievability results  as well as thorough discussions on the optimal choice of the relays' compressions, are presented, under the two different frameworks of block-by-block forward decoding and decoding after all the blocks have been finished, respectively. Finally, some concluding remarks are included in Section \ref{conclusion}.


\section{Main Results}
\label{S:mainresults}
Consider the multiple-relay channel depicted in Fig. \ref{fig2}, which can be denoted by
$$
(\XX\times \XX_1\times\cdots\times \XX_n,\,p(y,y_1,\ldots,y_n|x,x_1,\ldots,x_n),\,\YY \times \YY_1\times \cdots \times \YY_n)
$$
where, $\XX,\XX_1,\ldots,\XX_n$
are the transmitter alphabets of the source and the relays respectively,
$\YY,\YY_1,\ldots,\YY_n$ are the receiver alphabets of the destination and
the relays respectively, and a collection of probability distributions
$p(\cdot,\cdot,\ldots,\cdot|x,x_1,\ldots,x_n)$ on $\YY\times \YY_1\times \cdots \times \YY_n$, one for each
$(x,x_1,\ldots,x_n)\in \XX\times \XX_1\times\cdots\times \XX_n$. The interpretation is that $x$
is the input to the channel from the source, $y$ is
the output of the channel to the destination, and $y_i$ is
the output received by the $i$-th relay. The $i$-th relay sends an input $x_i$ based on what it has received:
\begin{equation}
\label{processing2}
x_i(t)=f_{i,t}(y_i(t-1),y_i(t-2),\ldots), ~~ \mbox{ for every time } t,
\end{equation}
where $f_{i,t}(\cdot)$ can be any causal function.

Before presenting the main results, we introduce some simplified notations. Denote the set $\NN=\{1,2,\ldots,n\}$, and for any subset $\SS\subseteq \NN$, let $X_{\SS}=\{X_i,i\in \SS\}$, and use similar notations for other variables. The main results of the paper are presented in the following two different decoding frameworks: i) Block-by-block forward decoding; ii) Decoding after all the blocks have been finished, which includes all blocks united decoding and block-by-block backward decoding.

\subsection{Block-by-block forward decoding}
 Under the block-by-block forward decoding framework, the achievable rate with successive compression-message decoding and the achievable rate with joint compression-message decoding are presented in Theorems \ref{cfs} and \ref{cfj} respectively. Then the optimality of successive decoding is stated in Theorem \ref{optcf}, and it is shown that the optimal rate can be achieved only if the compressions at the relays are chosen such that they can be first decoded at the destination, i.e., successive compression-message decoding can be carried out. All the related proofs are presented in Section \ref{S:B-Bforward}.
\begin{theorem}
\label{cfs}
For the multiple-relay channel depicted in Fig. \ref{fig2}, by the cumulative encoding/block-by-block forward decoding/compression-message successive decoding scheme, a rate $R_{\text{C/F/S}}$ is achievable if for some
$$
p(x)p(x_1)\cdots p(x_n)p(\hat y_1|y_1,x_1)\cdots p(\hat y_n|y_n,x_n),
$$
there exists a rate vector $\{R_i,i=1,\ldots,n\}$ satisfying
\begin{equation}
\label{cfsmac}
\sum_{i\in \SS_1}R_i \leq I(X_{\SS_1};Y|X_{\SS_1^c})
\end{equation}
for any subset ${{\cal S}_1}\subseteq {\cal N}$, such that for
any subset ${\cal S}\subseteq {\cal N}$,
\begin{equation}
\label{cfscondition}
I(Y_\SS; \hat Y_\SS|\hat Y_{\SS^c},Y,X_\NN) \leq \sum_{i\in \SS}R_i
\end{equation}
and
\begin{equation}
\label{cfsrate}
R_{\text{C/F/S}}<I(X;\hat Y_\NN, Y|X_\NN).
\end{equation}
\end{theorem}

\begin{theorem}
\label{cfj}
For the multiple-relay channel depicted in Fig. \ref{fig2}, by the cumulative encoding/block-by-block forward decoding/compression-message joint decoding scheme, a rate $R_{\text{C/F/J}}$ is achievable if for some
$$
p(x)p(x_1)\cdots p(x_n)p(\hat y_1|y_1,x_1)\cdots p(\hat y_n|y_n,x_n),
$$
there exists a rate vector $\{R_i,i=1,\ldots,n\}$ satisfying
\begin{equation}
\label{cfjmac}
\sum_{i\in \SS_1}R_i \leq I(X_{\SS_1};Y|X_{\SS_1^c})
\end{equation}
for any subset ${{\cal S}_1}\subseteq {\cal N}$, such that for
any subset ${\cal S}\subseteq {\cal N}$,
\begin{equation}
\label{cfjrate}
R_{\text{C/F/J}}<I(X;\hat Y_\NN, Y|X_\NN)-I(Y_\SS; \hat Y_\SS|\hat Y_{\SS^c},Y,X_\NN)+\sum_{i\in \SS}R_i.
\end{equation}
\end{theorem}

Let $R_{\text{C/F/S}}^*$ and $R_{\text{C/F/J}}^*$ be the supremum of the achievable rates stated in Theorems \ref{cfs} and \ref{cfj} respectively.

\begin{theorem}
\label{optcf}
$R_{\text{C/F/S}}^*=R_{\text{C/F/J}}^*$, and $R_{\text{C/F/J}}^*$ can be obtained only when the distribution
$$p(x)p(x_1)\cdots p(x_n)p(\hat y_1|y_1,x_1)\cdots p(\hat y_n|y_n,x_n)$$ is chosen such that there exists a rate vector $\{R_i,i=1,\ldots,n\}$ satisfying \dref{cfsmac}-\dref{cfscondition}.
\end{theorem}

\subsection{Decoding after all the blocks have been finished}

It was shown in \cite{KimElGamal} that the original cumulative encoding/block-by-block forward decoding/compression-message successive decoding scheme developed in \cite{covelg79} can be improved to achieve higher rates in the case of multiple relays, although no improvement was obtained in the case of a single relay. In their new compress-and-forward relay scheme \cite{KimElGamal}, cumulative encoding was replaced by repetitive encoding, and block-by-block forward decoding was replaced by all blocks united decoding. They also used joint instead of successive compression-message decoding. For the single-source multiple-relay channel depicted in Fig. \ref{fig2}, their Theorem 1 in \cite{KimElGamal} can be re-stated as the following theorem.
\begin{theorem}
\label{T:ruj}
For the multiple-relay channel depicted in Fig. \ref{fig2}, a rate $R_{\text{R/U/J}}$ is achievable if there exists some
$$
p(x)p(x_1)\cdots p(x_n)p(\hat y_1|y_1,x_1)\cdots p(\hat y_n|y_n,x_n),
$$
such that
\begin{equation}
\label{rujrate}
R_{\text{R/U/J}}<\min_{\SS \subseteq \NN}I(X,X_\SS;\hat Y_{\SS^c},Y|X_{\SS^c})-I(Y_{\SS};\hat Y_{\SS}|X,X_\NN,Y, \hat Y_{\SS^c}).
\end{equation}
\end{theorem}

In this paper, we will show that the improvement is not a result of replacing cumulative encoding by repetitive encoding, but actually, is a benefit obtained when the decoding is delayed, i.e., only starts after all the blocks have been finished. Besides all blocks united decoding, we will show that block-by-block backward decoding also achieves the same improvement since it also starts the decoding after all the blocks have been finished.

Similar to the framework of block-by-block forward decoding, we will also show that for these new schemes with decoding after all the blocks have been finished, the optimal rate can be achieved only when the compressions at the relays are chosen such that successive compression-message decoding can be carried out. Thus, in terms of complexity, cumulative encoding/block-by-block backward decoding/compression-message successive decoding is the simplest choice in achieving the highest rate in the case of multiple relays. The corresponding achievable rate is presented in the following theorem.
\begin{theorem}
\label{T:cbs}
For the multiple-relay channel depicted in Fig. \ref{fig2}, a rate $R_{\text{C/B/S}}$ is achievable if there exists some
$$
p(x)p(x_1)\cdots p(x_n)p(\hat y_1|y_1,x_1)\cdots p(\hat y_n|y_n,x_n),
$$
such that for any subset ${\cal S}\subseteq {\cal N}$,
\begin{equation}
\label{cbscondition}
I(X_\SS; \hat Y_{\SS^c},Y|X_{\SS^c}) -  I(Y_{\SS};\hat Y_{\SS}|X_\NN,Y, \hat Y_{\SS^c})\geq 0,
\end{equation}
and
\begin{equation}
\label{cbsrate}
R_{\text{C/B/S}}<I(X;\hat Y_\NN, Y|X_\NN).
\end{equation}
\end{theorem}

Let $R_{\text{R/U/J}}^*$ and $R_{\text{C/B/S}}^*$ be the supremum of the achievable rates stated in Theorem \ref{T:ruj} and \ref{T:cbs} respectively, i.e.,
\begin{align*}
R_{\text{R/U/J}}^*:=&\max_{p(x)\prod_{i=1}^{n}p(x_i)p(\hat y_i|x_i,y_i)}\min_{\SS \subseteq \NN}I(X,X_\SS;\hat Y_{\SS^c},Y|X_{\SS^c})-I(Y_{\SS};\hat Y_{\SS}|X,X_\NN,Y, \hat Y_{\SS^c}).
\end{align*}
and
\begin{align}
R_{\text{C/B/S}}^*:= &\max_{p(x)\prod_{i=1}^{n}p(x_i)p(\hat y_i|x_i,y_i)} I(X;\hat Y_\NN, Y|X_\NN) \nonumber \\
\text{such that~} &I(X_\SS; \hat Y_{\SS^c},Y|X_{\SS^c}) -  I(Y_{\SS};\hat Y_{\SS}|X_\NN,Y, \hat Y_{\SS^c})\geq 0 , \forall \SS \subseteq \NN. \label{eq1}
\end{align}
The optimality of successive decoding is demonstrated in the following theorem.
\begin{theorem}
\label{T:optallblocks}
$R_{\text{R/U/J}}^*=R_{\text{C/B/S}}^*$, and $R_{\text{R/U/J}}^*$ can be obtained only when the distribution
$$
p(x)p(x_1)\cdots p(x_n)p(\hat y_1|y_1,x_1)\cdots p(\hat y_n|y_n,x_n)
$$
is chosen such that \dref{eq1} holds.
\end{theorem}

\vskip 0.5cm

As mentioned in the Introduction, although the optimal rate is achieved only when successive decoding can be supported, there are situations where it is of interest to consider other compressions not supporting successive decoding. Hence, more generally, we will use the cumulative encoding/block-by-block backward decoding/compression-message joint decoding. The corresponding achievable rate is given in the following theorem.
\begin{theorem}
\label{T:cbj}
For the multiple-relay channel depicted in Fig. \ref{fig2}, with a given distribution
$$
p(x)p(x_1)\cdots p(x_n)p(\hat y_1|y_1,x_1)\cdots p(\hat y_n|y_n,x_n),
$$
a rate $R_{\text{C/B/J}}$
is achievable if
\begin{equation}
\label{cbjrate}
R_{\text{C/B/J}}
<\min_{\SS \subseteq \DD_{\text{J}}}I(X,X_\SS;\hat Y_{\DD_{\text{J}} \setminus \SS},Y|X_{\DD_{\text{J}} \setminus \SS})-I(Y_{\SS};\hat Y_{\SS}|X,X_{\DD_{\text{J}}},Y, \hat Y_{\DD_{\text{J}} \setminus \SS}),
\end{equation}
where $\DD_{\text{J}}$ is the unique largest subset of $\NN$ satisfying
\begin{equation}
\label{eq2}
I(X_\SS;\hat Y_{\DD_{\text{J}} \setminus \SS},Y|X,X_{\DD_{\text{J}} \setminus \SS})-I(Y_{\SS};\hat Y_{\SS}|X,X_{\DD_{\text{J}}},Y, \hat Y_{\DD_{\text{J}} \setminus \SS})> 0,
\end{equation}
for any nonempty $\SS \subseteq \DD_{\text{J}}$. In addition, $\hat Y_{\DD_{\text{J}}}$ can be decoded jointly with $X$.
\end{theorem}

There also exists a unique largest subset $\DD'_{\text{J}}\subseteq \NN$ satisfying
\begin{equation}
\label{eq3}
I(X_\SS;\hat Y_{\DD'_{\text{J}} \setminus \SS},Y|X,X_{\DD'_{\text{J}} \setminus \SS})-I(Y_{\SS};\hat Y_{\SS}|X,X_{\DD'_{\text{J}}},Y, \hat Y_{\DD'_{\text{J}} \setminus \SS})\geq 0,
\end{equation}
for any $\SS \subseteq \DD'_{\text{J}}$. It will be clear from the proof of Theorem \ref{T:cbj} that the compressions of the relays in $\NN \setminus \DD'_{\text{J}}$ are not decodable even jointly with the message.

\vskip 0.5cm

On the other hand, the achievable rate \dref{rujrate} can be more generally expressed as
\begin{equation}
\label{rujimproved}
R_{\text{R/U/J}}<\min_{\SS \subseteq \MM}I(X,X_\SS;\hat Y_{\MM \setminus \SS},Y|X_{\MM \setminus \SS})-I(Y_{\SS};\hat Y_{\SS}|X,X_\MM,Y, \hat Y_{\MM \setminus \SS})
\end{equation}
if we only consider a subset of relays $\MM\subseteq \NN$ for the decoding, while treating the other inputs as purely noise. Interestingly, the following theorem implies that $\MM=\NN$ may not be the optimal choice to maximize the R.H.S. (right-hand-side) of \dref{rujimproved}, i.e., sometimes, it is better to consider only a subset of relays.
\begin{theorem}
\label{T:equi}
For any $p(x)\prod_{i=1}^{n}p(x_i)p(\hat y_i|x_i,y_i)$, among all the choices of $\MM\subseteq \NN$,
the R.H.S. of \dref{rujimproved} is maximized when $\MM=\DD_{\text{J}}$ or $\MM=\DD'_{\text{J}}$, but is strictly less than the maximum when $\MM \nsubseteq \DD'_{\text{J}}$.
Here, $\DD_{\text{J}}$ and $\DD'_{\text{J}}$ are defined as in \dref{eq2} and \dref{eq3}.
\end{theorem}

Therefore, not only the compressions of the relays in  $\NN \setminus \DD'_{\text{J}}$ are not decodable, but also including them in the formula \dref{rujimproved}, i.e., choosing $\MM \nsubseteq \DD'_{\text{J}}$, will even strictly lower the achievable rate.

By comparing \dref{cbjrate} and \dref{rujimproved} with $\MM=\DD_{\text{J}}$, Theorem \ref{T:equi} also implies that for any compressions chosen at the relays, the cumulative encoding/block-by-block backward decoding/compression-message joint decoding scheme achieves the same rate as the repetitive encoding/all blocks united decoding/compression-message joint decoding scheme.

The proofs of Theorems \ref{T:cbs}-\ref{T:equi} are presented in Section \ref{S:DecodingAllBlocks}.

\section{Block-by-Block Forward Decoding}
\label{S:B-Bforward}
We first prove the achievability results stated in Theorems \ref{cfs} and \ref{cfj} respectively.

In both the cumulative encoding/block-by-block forward decoding/compression-message successive decoding and the cumulative encoding/block-by-block forward decoding/compression-message joint decoding schemes, the codebook generation and encoding processes are exactly the same as the classical way, i.e., the way in the proof of Theorem 6 of \cite{covelg79}. The difference between these two schemes is only on the decoding process at the destination: i) In successive decoding, the destination first finds, from the specific bins sent by the relays via $X_1, X_2,\ldots,X_n$, the unique combination of $\hat Y_1, \hat Y_2,\ldots,\hat Y_n$ sequences that is jointly typical with the $Y$ sequence received, and then finds the unique $X$ sequence that is jointly typical with the $Y$ sequence received, and also with the previously recovered $\hat Y_1, \hat Y_2,\ldots,\hat Y_n$ sequences. ii) In joint decoding, the destination finds the unique $X$ sequence that is jointly typical with the $Y$ sequence received, and also with some combination of $\hat Y_1, \hat Y_2,\ldots,\hat Y_n$ sequences from the specific bins sent by the relays via $X_1, X_2,\ldots,X_n$.

\subsection{A simplified model and proof of Theorem \ref{cfs}}
To make the presentation easier to follow, we introduce a simplified channel model as depicted in Fig. \ref{fig4}, where, the relays are connected to the destination via error-free digital links with capacities $R_1,R_2,\ldots,R_n$, where $(R_1,R_2,\ldots,R_n)$ are chosen based on \dref{cfsmac}. The $i$-th digital link plays the same role as the $X_i\rightarrow Y$ link in Fig. \ref{fig2}, for any $i=1,2,\ldots,n$. Such a replacement will not lead to any essential variation of the original coding scheme, since under the original coding framework, the $X_i\rightarrow Y$ link is used as a separate link to forward digital information. The benefit of directly replacing it by a digital link is that the codebook construction for $\hat Y_i$ can be simplified, since no $X_i$ needs to be considered. For this simplified model, \dref{cfscondition} and \dref{cfsrate} simplify to
\begin{equation}
\label{cfsconditionsimplified}
I(Y_\SS; \hat Y_\SS|\hat Y_{\SS^c},Y) \leq \sum_{i\in \SS}R_i
\end{equation}
and
\begin{equation}
\label{cfsratesimplified}
R_{\text{C/F/S}}<I(X;\hat Y_\NN, Y).
\end{equation}

\begin{figure}[hbt]
\centering
\includegraphics[width=2.6in]{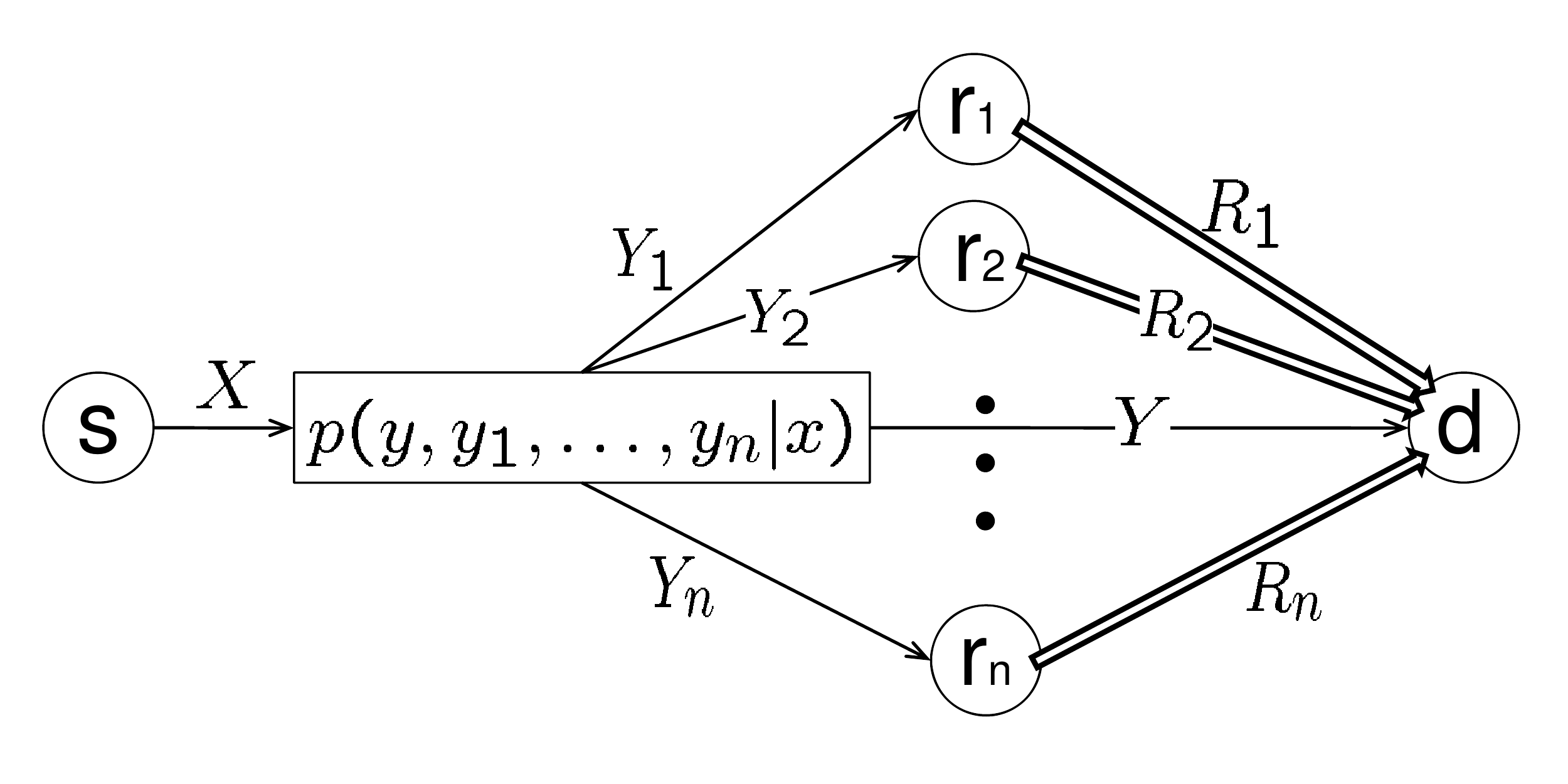}
\caption{A simplified multiple-relay model with digital links.}
 \label{fig4}
\end{figure}

The basic idea of the compress-and-forward strategy is for the relay to compress its observations into some approximations, which can be represented by fewer number of bits, and thus, can be forwarded to the destination. To deal with delay at the relay, block Markov coding was used, where the total time is divided into a sequence of blocks of equal length $T$, and coding is performed block by block. For example, each relay compresses its observations of each block at the end of the block, and forwards the approximations in the next block. Therefore, to decode the message sent by the source in any block, it is not until the end of the next block, has the destination received the help from the relay.

The encoding process is exactly the same as that in the proof of Theorem 6 of \cite{covelg79}. We only emphasize that the $i$-th relay needs to generate $2^{T(I(Y_i;\hat Y_i)+\epsilon)}$ many $\hat Y_i$ sequences, and randomly throws them into $2^{TR_i}$ bins. At the end of each block, the relay finds a $\hat Y_i$ sequence which is jointly typical with the $Y_i$ sequence it received during the block, and in the next block, informs the destination the index of the bin that contains the $\hat Y_i$ sequence.

The decoding process operates in a successive way. At the end of each block $b=2,3,\ldots$, the destination first finds, from the bins forwarded by the relays during block $b$, the unique combination of $\hat Y_1,\hat Y_2,\ldots,\hat Y_n$ sequences that is jointly typical with the $Y$ sequence received,  i.e.,
\begin{align}
\label{cfstypicalcheck}
 (\underline{\hat Y}_{1}(b-1),\ldots,\underline{\hat Y}_{n}(b-1),\underline{Y}(b-1)) \in A_\epsilon(\hat Y_\NN,Y).
\end{align}

Error occurs if the true $\underline{\hat Y}_{\NN}(b-1)$ does not satisfy \dref{cfstypicalcheck}, or a false $\underline{\hat Y}_{\NN}(b-1)$ satisfies \dref{cfstypicalcheck}. According to the properties of typical sequences, the true $\underline{\hat Y}_{\NN}(b-1)$ satisfies \dref{cfstypicalcheck} with high probability.

The probability of a false $\underline{\hat Y}_{\NN}(b-1)$ with some false $\{\underline{\hat Y}_i(b-1),\,i\in \SS\}$ but true $\{\underline{\hat Y}_i(b-1),\,i\in \SS^c\}$ being jointly typical with $\underline{Y}(b-1)$ can be upper bounded by
$$
2^{T(H(Y,\hat Y_\NN)+\epsilon)}2^{-T(H(Y,\hat Y_{\SS^c})-\epsilon)}\prod_{i\in \SS}2^{-T(H(\hat Y_i)-\epsilon)}.
$$
There are $\prod_{i\in \SS}(2^{T(I(Y_i;\hat Y_i)-R_i+\epsilon)}-1 )$ false $\underline{\hat Y}_\SS (b-1)$ from the bins, thus the probability of finding such a false $\underline{\hat Y}_{\NN}(b-1)$ can be upper bounded by
\begin{align*}
2^{T(H(Y,\hat Y_\NN)+\epsilon)}2^{-T(H(Y,\hat Y_{\SS^c})-\epsilon)}\prod_{i\in \SS}2^{-T(H(\hat Y_i)-I(Y_i;\hat Y_i)+R_i-2\epsilon)},
\end{align*}
which tends to zero for sufficiently small $\epsilon$ as $T\rightarrow \infty$, if
\begin{align}
\label{cfsneedtosim}
H(\hat Y_{\SS}|Y,\hat Y_{\SS^c})- \sum_{i\in \SS}[H(\hat Y_i|Y_i)+R_i]<0.
\end{align}
Leting $\SS=\{i_j\in \NN: j=1,\ldots,|\SS|\}$, we have
\begin{align*}
\sum_{i\in \SS}H(\hat Y_i|Y_i)=&\sum_{j=1,\ldots,|\SS|}H(\hat Y_{i_j}|Y_{i_j})\\
=&\sum_{j=1,\ldots,|\SS|}H(\hat Y_{i_j}|Y_{\SS},Y,\hat Y_{\SS^c}, \{\hat{Y}_{i_1},\ldots,\hat{Y}_{i_{j-1}}\})\\
=&H(\hat Y_{\SS}|Y_{\SS},Y,\hat Y_{\SS^c}).
\end{align*}
Plugging this into \dref{cfsneedtosim}, we obtain \dref{cfsconditionsimplified}\footnote{The case of ``$=$'' can be included since \dref{cfsratesimplified} doesn't include ``$=$''. The same consideration applies throughout the paper.}.

Given that \dref{cfsconditionsimplified} is satisfied, the destination can recover $\underline{\hat Y}_{\NN}(b-1)$ at the end of block $b$. Then, based on $\underline{\hat Y}_{\NN}(b-1)$ and $\underline{Y}(b-1)$, $\underline{X}(w)$ can be recovered if \dref{cfsratesimplified} holds.

\subsection{Proof of Theorem \ref{cfj}}

Similarly, we consider the simplified model as depicted in Fig. \ref{fig4}, where the rates $(R_1,R_2,\ldots,R_n)$ are chosen based on \dref{cfjmac}. Then, \dref{cfjrate} simplifies to
\begin{equation}
\label{rjratesimplified}
R_{\text{C/F/J}}<I(X;\hat Y_\NN, Y)-I(Y_\SS; \hat Y_\SS|\hat Y_{\SS^c},Y)+\sum_{i\in \SS}R_i.
\end{equation}

In cumulative encoding/block-by-block forward decoding/compression-message joint decoding, the encoding part is exactly the same as that in the proof of Theorem \ref{cfs}, and the decoding process operates as the following. At the end of each block $b=2,3,\ldots$, the destination finds the unique $X$ sequence that is jointly typical with the $Y$ sequence received during block $b-1$, and also with some $\hat Y_1,\hat Y_2,\ldots,\hat Y_n$ sequences from the bins forwarded by the relays during block $b$, i.e.,
\begin{align}
\label{typm}
 (\underline{X}(w),\underline{Y}(b-1),\underline{\hat Y}_{\NN}(b-1)) \in A_\epsilon(X,Y,\hat Y_{\NN}).
\end{align}

Error occurs if the true $\underline{X}(w)$ does not satisfy \dref{typm}, or a false $\underline{X}(w')$ satisfies \dref{typm}. According to the properties of typical sequences, the true $\underline{X}(w)$ satisfies \dref{typm} with high probability.

The probability of a false $\underline{X}(w')$ being jointly typical with $\underline{Y}(b-1)$ and some false $\{\underline{\hat Y}_i(b-1),\,i\in \SS\}$ but true $\{\underline{\hat Y}_i(b-1),\,i\in \SS^c\}$ can be upper bounded by
$$
2^{T(H(X,Y,\hat Y_\NN)+\epsilon)}2^{-T(H(X)-\epsilon)}2^{-T(H(Y,\hat Y_{\SS^c})-\epsilon)}\prod_{i\in \SS}2^{-T(H(\hat Y_i)-\epsilon)}.
$$
There are $2^{TR}-1$ false $w'$, and $\prod_{i\in \SS}(2^{T(I(Y_i;\hat Y_i)-R_i+\epsilon)}-1 )$ false $\underline{\hat Y}_\SS (b-1)$ from the bins, thus the probability of finding such a false $\underline{X}(w')$ can be upper bounded by
\begin{align*}
&2^{TR} 2^{T(H(X,Y,\hat Y_\NN)+\epsilon)}2^{-T(H(X)-\epsilon)}\nonumber \\
\times~&2^{-T(H(Y,\hat Y_{\SS^c})-\epsilon)}\prod_{i\in \SS}2^{-T(H(\hat Y_i)-I(Y_i;\hat Y_i)+R_i-2\epsilon)},
\end{align*}
which tends to zero for sufficiently small $\epsilon$ as $T\rightarrow \infty$, if \dref{rjratesimplified} holds.

\subsection{Optimality of successive decoding in block-by-block forward decoding}\label{Sub:optcf}

To make the proof of Theorem \ref{optcf} easier to follow, we still consider the simplified model depicted in Fig. \ref{fig4}. Then, $R_{\text{C/F/S}}^*$ and $R_{\text{C/F/J}}^*$ can be respectively written as
\begin{align}
R_{\text{C/F/S}}^*= &\max_{p(x)\prod_{i=1}^{n}p(\hat y_i|y_i)} I(X;\hat Y_\NN, Y) \label{rsmax} \\
\text{such that~} &I(Y_\SS; \hat Y_\SS|\hat Y_{\SS^c},Y)- \sum_{i\in \SS}R_i\leq 0 , \forall \SS \subseteq \NN,\label{rsmaxconstraint}
\end{align}
and
\begin{align}
R_{\text{C/F/J}}^*=&\max_{p(x)\prod_{i=1}^{n}p(\hat y_i|y_i)}\min_{\SS \subseteq \NN}\{I(X;\hat Y_\NN, Y)-I(Y_\SS; \hat Y_\SS|\hat Y_{\SS^c},Y)+\sum_{i\in \SS}R_i    \}.\label{rjmax}
\end{align}

Before proceeding to the proof of Theorem \ref{optcf}, we first introduce some useful notations and lemmas. Let
\begin{align}
I_{\mathcal{A},\BB}(\SS):=&\sum_{i\in \SS}R_i-I(Y_\SS; \hat Y_\SS|\hat Y_{\mathcal{A}}, \hat Y_{\BB \setminus \SS},Y), \forall \SS \subseteq \BB,\\
I_{\mathcal{B}}(\SS):=&I_{\emptyset,\BB}(\SS)=\sum_{i\in \SS}R_i-I(Y_\SS; \hat Y_\SS|\hat Y_{\mathcal{B} \setminus \SS},Y), \forall \SS \subseteq \mathcal{B}, \\
I(\SS):=&I_{\NN}(\SS)=\sum_{i\in \SS}R_i-I(Y_\SS; \hat Y_\SS|\hat Y_{\SS^c},Y), \forall \SS \subseteq \NN.
\end{align}
Then, we have the following lemmas, whose proofs are given in Appendix \ref{A:A}.
\begin{lemma}
\label{L:union}
1) If $I_{\mathcal{A}}(\SS_1)\geq 0$, $\forall \SS_1 \subseteq \mathcal{A}$, and $I_{\mathcal{B}}(\SS_2)\geq 0$, $\forall \SS_2 \subseteq \mathcal{B}$, then $I_{\mathcal{A}\cup \BB}(\SS)\geq 0$, $\forall \SS \subseteq \mathcal{A}\cup \BB.$

2) If $I_{\mathcal{A}}(\SS_1)\geq 0$, $\forall \SS_1 \subseteq \mathcal{A}$, and $I_{\mathcal{A},\mathcal{B}}(\SS_2)\geq 0$, $\forall \SS_2 \subseteq \mathcal{B}$, then $I_{\mathcal{A}\cup \BB}(\SS)\geq 0$, $\forall \SS \subseteq \mathcal{A}\cup \BB.$
\end{lemma}

\begin{lemma}
\label{L:largestD}
Under any $p(x)\prod_{i=1}^{n}p(\hat y_i|y_i)$, there exists a unique set $\DD$, which is the largest subset of $\NN$ satisfying
$$I_{\DD}(\SS)\geq 0, \forall \SS \subseteq \DD.$$
\end{lemma}

\begin{lemma}
\label{L:exsit}
If $I_{\mathcal{A},\BB}(\BB)\geq 0$ for some nonempty $\BB$, then there exists some nonempty $\mathcal{C} \subseteq \BB$ such that $I_{\mathcal{A},\CC}(\SS)\geq 0, \forall \SS \subseteq \mathcal{C}$.
\end{lemma}

\begin{lemma}
\label{L:aub}
For any $\mathcal{A}$ and $\BB$ with $\mathcal{A} \cap \BB=\emptyset$, $I(\mathcal{A})+I(\mathcal{B})=I(\mathcal{A} \cup \BB)+I(\hat Y_{\mathcal{A}};\hat Y_{\mathcal{B}}|\hat Y_{(\mathcal{A} \cup \BB)^c},Y).$
\end{lemma}

We are now ready to prove Theorem \ref{optcf}.

\begin{proof}[Proof of Theorem \ref{optcf}]
We show $R_{\text{C/F/S}}^* = R_{\text{C/F/J}}^*$ by showing that $R_{\text{C/F/S}}^* \leq R_{\text{C/F/J}}^*$ and $R_{\text{C/F/S}}^* \geq R_{\text{C/F/J}}^*$ respectively. Under any $p(x)\prod_{i=1}^{n}p(\hat y_i|y_i)$ such that $I(Y_\SS; \hat Y_\SS|\hat Y_{\SS^c},Y) \leq \sum_{i\in \SS}R_i$, $\forall \SS \subseteq \NN$, we have
$$ \min_{\SS \subseteq \NN}\{I(X;\hat Y_\NN, Y)-I(Y_\SS; \hat Y_\SS|\hat Y_{\SS^c},Y)+\sum_{i\in \SS}R_i    \} =I(X;\hat Y_\NN, Y), $$
and thus $R_{\text{C/F/S}}^* \leq R_{\text{C/F/J}}^*$.

To show $R_{\text{C/F/S}}^* \geq R_{\text{C/F/J}}^*$, it is sufficient to show that $R_{\text{C/F/J}}^*$ can be achieved only with $p(x)\prod_{i=1}^{n}p(\hat y_i|y_i)$ such that $I(\SS)\geq 0$, $\forall \SS \subseteq \NN$.
We will show this by two steps as follows: i) We first show that under any $p(x)\prod_{i=1}^{n}p(\hat y_i|y_i)$, if $\DD^c \neq \emptyset$, then $\DD^c \in \argmin_{\SS \subseteq \NN}I(\SS)$ and $\bigcap_{\TT \in \argmin_{\SS \subseteq \NN}I(\SS)}\TT=\DD^c$, where $\DD$ is defined as in Lemma \ref{L:largestD} and $\argmin_{\SS \subseteq \NN}I(\SS):=\{\TT \subseteq \NN: I(\TT)=\min_{\SS \subseteq \NN}I(\SS)\}$.   ii) We then argue that, under the optimal $p(x)\prod_{i=1}^{n}p(\hat y_i|y_i)$, $\DD^c$ must be $\emptyset$, i.e., $\DD$ must be $\NN$, and thus by the definition of $\DD$, $I(\SS)\geq0, \forall \SS \subseteq \NN$.

i) Assuming $\DD^c \neq \emptyset$ throughout Part i), we show $\DD^c \in \argmin_{\SS \subseteq \NN}I(\SS)$ and $\bigcap_{\TT \in \argmin_{\SS \subseteq \NN}I(\SS)}\TT=\DD^c$.

1) We first show $I(\DD^c)<0$ by using a contradiction argument. Suppose $I(\DD^c) \geq 0$, i.e., $I_{\DD,\DD^c}(\DD^c) \geq 0$. Then, by  Lemma \ref{L:exsit}, we have that there exists some nonempty $\BB \subseteq \DD^c$ such that $I_{\DD,\BB}(\SS) \geq 0$, $\forall \SS \subseteq \BB$. This will further imply, by Part 2) of Lemma \ref{L:union}, that $I_{\DD \cup \BB}(\SS)\geq 0, \forall \SS \subseteq \DD \cup \BB$. This is contradictory with the definition of $\DD$, and thus $I(\DD^c)<0$.

2) We show that $\forall \mathcal{A} \subseteq \DD^c$ and $\mathcal{A} \neq \DD^c$, $I(\mathcal{A})>I(\DD^c)$, and thus $I(\mathcal{A})>\min_{\SS \subseteq \NN}I(\mathcal{S})$. The proof is still by contradiction. Suppose that there exists some $\mathcal{A} \subseteq \DD^c$ and $\mathcal{A} \neq \DD^c$ such that $I(\mathcal{A}) \leq I(\DD^c)$. Then $I(\DD^c)-I(\mathcal{A})\geq 0$, i.e.,
\begin{align*}
&\sum_{i\in \DD^c} R_i - I(Y_{\DD^c};\hat Y_{\DD^c}|\hat Y_{\DD},Y)-\sum_{i\in \mathcal{A}} R_i + I(Y_{\mathcal{A}};\hat Y_{\mathcal{A}}|\hat Y_{\mathcal{A}^c},Y)\\
=& \sum_{i\in \DD^c\setminus \mathcal{A}} R_i - I(Y_{\DD^c\setminus \mathcal{A}};\hat Y_{\DD^c\setminus \mathcal{A}}|\hat Y_{\DD},Y)\\
=& I_{\DD,\DD^c\setminus \mathcal{A}}(\DD^c\setminus \mathcal{A})\\
\geq & 0.
\end{align*}
Again by Lemma \ref{L:exsit} and \ref{L:union} successively, we can conclude that there exists some nonempty $\BB \subseteq \DD^c\setminus \mathcal{A}$, such that $I_{\DD \cup \BB}(\SS)\geq 0, \forall \SS \subseteq \DD \cup \BB$, which is in contradiction. Therefore, $I(\mathcal{A})>I(\DD^c)\geq \min_{\SS \subseteq \NN}I(\mathcal{S})$.

3) We prove that $\forall \mathcal{A}$ with $\mathcal{A}\DD \neq \emptyset$ and $\mathcal{A}\DD^c \neq \DD^c$, $I(\mathcal{A})>\min_{\SS \subseteq \NN}I(\mathcal{S})$. Let $\mathcal{A}_1=\mathcal{A}\DD$ and $\mathcal{A}_2=\mathcal{A}\DD^c$. Then, we have, by Lemma \ref{L:aub}, that
\begin{align*}
I(\mathcal{A})=&I(\mathcal{A}_1 \cup \mathcal{A}_2)=I(\mathcal{A}_1)+I(\mathcal{A}_2)-I(\hat Y_{\mathcal{A}_1};\hat Y_{\mathcal{A}_2}|\hat Y_{\mathcal{A}^c},Y),\\
I(\mathcal{A}_1 \cup \DD^c)=&I(\mathcal{A}_1)+I(\DD^c)-I(\hat Y_{\mathcal{A}_1};\hat Y_{\DD^c}|\hat Y_{(\mathcal{A}_1 \cup \DD^c)^c},Y).
\end{align*}
Since $I(\mathcal{A}_2)>I(\DD^c)$ by 2) and
\begin{align*}
&I(\hat Y_{\mathcal{A}_1};\hat Y_{\DD^c}|\hat Y_{(\mathcal{A}_1 \bigcup \DD^c)^c},Y)\\
=&I(\hat Y_{\mathcal{A}_1};\hat Y_{\DD^c \setminus \mathcal{A}_2}|\hat Y_{(\mathcal{A}_1 \bigcup \DD^c)^c},Y)+
I(\hat Y_{\mathcal{A}_1};\hat Y_{ \mathcal{A}_2}|\hat Y_{(\mathcal{A}_1 \bigcup \DD^c)^c},\hat Y_{\DD^c \setminus \mathcal{A}_2},Y)\\
=&I(\hat Y_{\mathcal{A}_1};\hat Y_{\mathcal{A}_2}|\hat Y_{\mathcal{A}^c},Y)+I(\hat Y_{\mathcal{A}_1};\hat Y_{\DD^c \setminus \mathcal{A}_2}|\hat Y_{(\mathcal{A}_1 \bigcup \DD^c)^c},Y)\\
\geq &I(\hat Y_{\mathcal{A}_1};\hat Y_{\mathcal{A}_2}|\hat Y_{\mathcal{A}^c},Y),
\end{align*}
we have $I(\mathcal{A})>I(\mathcal{A}_1 \cup \DD^c)\geq \min_{\SS \subseteq \NN}I(\mathcal{S})$.

4) We prove that $\forall \mathcal{A}$ with $\mathcal{A}\DD \neq \emptyset$ and $\mathcal{A}\DD^c = \DD^c$, $I(\mathcal{A})\geq I(\DD^c)$. Letting $\mathcal{A}_1=\mathcal{A}\DD$, we have
\begin{align*}
I(\mathcal{A})=&I(\mathcal{A}_1 \cup \DD^c)\\
=&I(\mathcal{A}_1)+I(\DD^c)-I(\hat Y_{\mathcal{A}_1};\hat Y_{\DD^c}|\hat Y_{(\mathcal{A}_1 \cup \DD^c)^c},Y)\\
=&\sum_{i\in \mathcal{A}_1}R_i-I( Y_{\mathcal{A}_1};\hat Y_{\mathcal{A}_1}|\hat Y_{\mathcal{A}^c_1},Y)-I(\hat Y_{\mathcal{A}_1};\hat Y_{\DD^c}|\hat Y_{(\mathcal{A}_1 \cup \DD^c)^c},Y)+I(\DD^c)\\
=&\sum_{i\in \mathcal{A}_1}R_i-I( \hat Y_{\mathcal{A}_1};\hat Y_{\DD^c}, Y_{\mathcal{A}_1}|\hat Y_{(\mathcal{A}_1 \cup \DD^c)^c},Y)+I(\DD^c)\\
=&\sum_{i\in \mathcal{A}_1}R_i-I( \hat Y_{\mathcal{A}_1}; Y_{\mathcal{A}_1}|\hat Y_{\DD \setminus \mathcal{A}_1},Y)+I(\DD^c)\\
=&I_{\DD}(\mathcal{A}_1)+I(\DD^c)\\
\geq &I(\DD^c).
\end{align*}

Combining 2)-4), we can conclude that $\DD^c \in \argmin_{\SS \subseteq \NN}I(\SS)$ and $\bigcap_{\TT \in \argmin_{\SS \subseteq \NN}I(\SS)}\TT=\DD^c$.

ii) We now argue that under the optimal $p(x)\prod_{i=1}^{n}p(\hat y_i|y_i)$ that achieves $R^*_{\text{C/F/J}}$, if $\DD^c \neq \emptyset$, then $R^*_{\text{C/F/J}}$ is not optimal; and hence $\DD^c$ must be $\emptyset$. The argument is extended from that in \cite{ElGamalKim} and the detailed analysis is as follows.

Suppose $\DD^c \neq \emptyset$ at the optimum. Then, $\DD^c \in \argmin_{\SS \subseteq \NN}I(\SS)$ and $\bigcap_{\TT \in \argmin_{\SS \subseteq \NN}I(\SS)}\TT=\DD^c$. Therefore,
\begin{align}
R^*_{\text{C/F/J}}=&I(X;\hat Y_{\NN}, Y)+I(\DD^c)\nonumber \\
=&I(X;\hat Y_{\DD},Y)+I(X;\hat Y_{\DD^c}|\hat Y_{\DD},Y)+\sum_{i\in \DD^c}R_i-I(X,Y_{\DD^c};\hat Y_{\DD^c}|\hat Y_{\DD},Y)\nonumber \\
=&I(X;\hat Y_{\DD},Y)+\sum_{i\in \DD^c}R_i-I(Y_{\DD^c};\hat Y_{\DD^c}|X,\hat Y_{\DD},Y),\label{E:seenf1}
\end{align}
and similarly,
\begin{align}
R^*_{\text{C/F/J}}=&I(X;\hat Y_{\NN}, Y)+I(\TT)\nonumber \\
=&I(X;\hat Y_{\TT^c},Y)+\sum_{i\in \TT}R_i-I(Y_{\TT};\hat Y_{\TT}|X,\hat Y_{\TT^c},Y),\label{E:seenf2}
\end{align}
for any $\TT \in \argmin_{\SS \subseteq \NN}I(\SS)$, $\TT \neq \DD^c$.

We argue that higher rate can be achieved. Consider $\hat Y'_1,\hat Y'_2,\ldots,\hat Y'_n$, where $\hat Y'_i=\hat Y_i$ for any $i\in \DD$, and $\hat Y'_i=\hat Y_i$ with probability $p$ and $\hat Y'_i=\emptyset$ with probability $1-p$ for any $i\in \DD^c$. When $p=1$, the achievable rate with $\hat Y'_1,\hat Y'_2,\ldots,\hat Y'_n$ is $R^*_{\text{C/F/J}}$. As $p$ decreases from 1, it can be seen from \dref{E:seenf1} and \dref{E:seenf2} that both $I(X;\hat Y'_{\NN}, Y)+I(\DD^c)$ and $I(X;\hat Y'_{\NN}, Y)+I(\TT)$ will increase, where $\TT \in \argmin_{\SS \subseteq \NN}I(\SS)$, $\TT \neq \DD^c$.  Thus, no matter how $I(X;\hat Y'_{\NN}, Y)+I(\SS)$  will change as $p$ decreases for $\SS \notin \argmin_{\SS \subseteq \NN}I(\SS)$, it is certain that there exists a $p^*$ such that the achievable rate by using $\hat Y'_1,\hat Y'_2,\ldots,\hat Y'_n$ is larger than $R^*_{\text{C/F/J}}$. This is in contradiction with the optimality of $R^*_{\text{C/F/J}}$, and thus at the optimum, $\DD^c$ must be $\emptyset$ , i.e., $I(\SS)\geq 0$, $\forall \SS \subseteq \NN$. This completes the proof of Theorem \ref{optcf}.
\end{proof}

\section{Decoding After All Blocks Have Been Finished}
\label{S:DecodingAllBlocks}
In this section, our discussion transfers to the compress-and-forward schemes with decoding after all blocks have been finished. The focus here is on the cumulative encoding/block-by-block backward decoding, since it is the simplest scheme to achieve the highest rate in the general multiple-relay channel, as mentioned before; for the repetitive encoding/all blocks united decoding, see the proof of Theorem 1 in \cite{KimElGamal}.

Cumulative encoding/block-by-block backward decoding can be combined with either compression-message successive decoding or compression-message joint decoding. In the following, we will first present the cumulative encoding/block-by-block backward decoding/compression-message successive decoding scheme to establish the achievable rate in Theorem \ref{T:cbs}, and demonstrate the optimality of successive decoding in the sense of Theorem \ref{T:optallblocks}. Then, the cumulative encoding/block-by-block backward decoding/compression-message joint decoding scheme will be used to prove Theorem \ref{T:cbj}, and the necessity of joint decodablity is demonstrated in the sense that only those relay nodes, whose compressions can be eventually decoded by joint decoding, are helpful to the decoding of the original message.

\subsection{Cumulative encoding/block-by-block backward decoding/compression-message successive decoding and Optimality of successive decoding}
In cumulative encoding/block-by-block backward decoding, the encoding process is similar to that in the proof of Theorem 6 in \cite{covelg79} (except that the binning at the relay is not needed here), but the decoding process operates backwardly. This scheme, combined with compression-message successive decoding, proves Theorem \ref{T:cbs} as follows.

\begin{proof}[Proof of Theorem \ref{T:cbs}]

\emph{Codebook Generation:} Fix $p(x)\prod_{i=1}^{n}p(x_i)p(\hat y_i|x_i,y_i)$. Consider $B+M$ blocks, where the source will
transmit information in the first $B$ blocks and keep silent in the last $M$ blocks, and $M \ll B$ such that the rate loss can be made arbitrarily small. We randomly and independently generate a codebook for each block.

For each block $b\in [1:B+M]$, randomly and independently generate $2^{TR_{\text{C/B/S}}}$ sequences $\mathbf{x}_b(m_b)$, $m_b \in [1:2^{TR_{\text{C/B/S}}}]$; for each block $b\in [1:B+M]$ and each relay node $i\in \NN$, randomly and independently generate $2^{T\hat R_i}$ sequences $\mathbf{x}_{i,b}(l_{i,b-1})$, $l_{i,b-1}\in [1:2^{T\hat R_i}]$, where $\hat R_i=I(Y_i;\hat Y_i|X_i)+\epsilon$; for each relay node $i\in \NN$ and each $\mathbf{x}_{i,b}(l_{i,b-1})$, $l_{i,b-1}\in [1:2^{T\hat R_i}]$, randomly and conditionally independently generate $2^{T\hat R_i}$ sequences $ \hat{\mathbf{y}}_{i,b}(l_{i,b}|l_{i,b-1})$, $l_{i,b}\in [1:2^{T\hat R_i}]$. This defines the codebook for any block $b\in [1:B+M]$,
$$
\mathcal{C}_b=\{\mathbf{x}_b(m_b), \mathbf{x}_{i,b}(l_{i,b-1}) , \hat{\mathbf{y}}_{i,b}(l_{i,b}|l_{i,b-1}): m_b \in [1:2^{TR_{\text{C/B/S}}}], l_{i,b},l_{i,b-1} \in  [1:2^{T\hat R_i}], i\in \NN          \}.
$$

\emph{Encoding:} Let $\mathbf{m}=(m_1,m_2,\ldots,m_B)$ be the message vector to be sent and let $m_b=1$ be the dummy message for any $b\in [B+1:B+M]$. For any block $b\in [1:B+M]$, each relay node $i\in \NN$, upon receiving $\mathbf{y}_{i,b}$ at the end of block $b$, finds an index $l_{i,b}$ such that $(\mathbf{x}_{i,b}(l_{i,b-1}), \mathbf{y}_{i,b},\hat{\mathbf{y}}_{i,b}(l_{i,b}|l_{i,b-1}))\in A_{\epsilon}(X_i,Y_i,\hat{Y}_i)$, where $l_{i,0}=1$ by convention. The codewords $ \mathbf{x}_b(m_b)$ and  $\mathbf{x}_{i,b}(l_{i,b-1}),i\in \NN   $ are transmitted in block $b$, $b\in [1:B+M]$.

\emph{Decoding:} i) The destination first finds a unique combination of the relays' compression indices $\mathbf{l}^B=(\mathbf{l}_1,\ldots,\mathbf{l}_B)$ and some $\mathbf{l}_{B+1}^{B+M}=(\mathbf{l}_{B+1},\ldots,\mathbf{l}_{B+M})$, where $\mathbf{l}_b=(l_{1,b},\ldots,l_{n,b})$, $\forall b\in [1:B+M]$, such that for any $b=1,\ldots,B+M$,
\begin{align}
\left(  (\mathbf{X}_{1,b}(l_{1,b-1}) , \hat{\mathbf{Y}}_{1,b}(l_{1,b}|l_{1,b-1})),\ldots,(\mathbf{X}_{n,b}(l_{n,b-1}) , \hat{\mathbf{Y}}_{n,b}(l_{n,b}|l_{n,b-1})), \mathbf{Y}_b   \right)\in A_{\epsilon}(X_\NN,\hat{Y}_\NN,Y).\label{E:check1}
\end{align}

Specifically, this can be done backwards as follows:

a) The destination finds the unique $\mathbf{l}_B$ such that there exists some $\mathbf{l}_{B+1}^{B+M}=(\mathbf{l}_{B+1},\ldots,\mathbf{l}_{B+M})$ satisfying (\ref{E:check1}) for any $b=B+1,\ldots,B+M$.

Assume the true $\mathbf{l}_B^{B+M}=\mathbf{1}^{M+1}$. Then, error occurs if $\mathbf{l}_B=\mathbf{1}$ does not satisfy \dref{E:check1} with any $\mathbf{l}_{B+1}^{B+M}$ for any $b=B+1,\ldots,B+M$, or a false  $\mathbf{l}_B \neq \mathbf{1}$ satisfies \dref{E:check1} with some $\mathbf{l}_{B+1}^{B+M}$ for any $b=B+1,\ldots,B+M$. Since $\mathbf{l}_B^{B+M}=\mathbf{1}^{M+1}$  satisfies \dref{E:check1} for any $b=B+1,\ldots,B+M$ with high probability according to the properties of typical sequences, we only need to bound $\mbox{Pr}(\bigcup_{\mathbf{l}_B \neq \mathbf{1}}\mathcal{E}_{\mathbf{l}_B})$, where $\mathcal{E}_{\mathbf{l}_B}$ is defined as the event that $\mathbf{l}_B$  satisfies \dref{E:check1} with some $\mathbf{l}_{B+1}^{B+M}$ for any $b=B+1,\ldots,B+M$. For any $(\mathbf{l}_{b-1},\mathbf{l}_b)$, define $\mathcal{A}_{b}(\mathbf{l}_{b-1},\mathbf{l}_b)$ as the event that $(\mathbf{l}_{b-1},\mathbf{l}_b)$ satisfies \dref{E:check1}. Then, we have
\begin{align}
\mbox{Pr}(\bigcup_{\mathbf{l}_B \neq \mathbf{1}}\mathcal{E}_{\mathbf{l}_B})=&\mbox{Pr}(\bigcup_{\mathbf{l}_{B+1}^{B+M}} \bigcup_{\mathbf{l}_B \neq \mathbf{1}} \bigcap_{b=B+1}^{B+M} \mathcal{A}_{b}(\mathbf{l}_{b-1},\mathbf{l}_b))\nonumber \\
= &\mbox{Pr}( \bigcup_{j=1}^{M-1} \bigcup_{\mathbf{l}_{B+1}^{B+M}: \mathbf{l}_{B+j}=\mathbf{1}}   \bigcup_{
\scriptsize \begin{array}{c}
\mathbf{l}_{B+1}^{B+M}: \mathbf{l}_{B+j}\neq \mathbf{1}, \\
\forall j\in[1:M-1]
\end{array}
} \bigcup_{\mathbf{l}_B \neq \mathbf{1}}\bigcap_{b=B+1}^{B+M} \mathcal{A}_{b}(\mathbf{l}_{b-1},\mathbf{l}_b))\nonumber \\
\leq & \sum_{j=1}^{M-1} \mbox{Pr}(\bigcup_{\mathbf{l}_{B+1}^{B+M}: \mathbf{l}_{B+j}=\mathbf{1}} \bigcup_{\mathbf{l}_B \neq \mathbf{1}} \bigcap_{b=B+1}^{B+M} \mathcal{A}_{b}(\mathbf{l}_{b-1},\mathbf{l}_b)) + \mbox{Pr}(\bigcup_{
\scriptsize \begin{array}{c}
\mathbf{l}_{B+1}^{B+M}: \mathbf{l}_{B+j}\neq \mathbf{1}, \\
\forall j\in[1:M-1]
\end{array}
} \bigcup_{\mathbf{l}_B \neq \mathbf{1}} \bigcap_{b=B+1}^{B+M} \mathcal{A}_{b}(\mathbf{l}_{b-1},\mathbf{l}_b)).\label{E:following}
\end{align}

Let us first consider the second term in \dref{E:following}.
For any $\mathbf{l}_B^{B+M}$, let $\SS_{b}(\mathbf{l}_B^{B+M})=\{i\in \NN: l_{i,b-1} \neq 1\}$. Note $\SS_{b}(\mathbf{l}_B^{B+M})$ only depends on $\mathbf{l}_{b-1}$, so we also write it as $\SS_{b}(\mathbf{l}_{b-1})$. Define $\mathbf{X}_b(\SS_{b}(\mathbf{l}_{b-1}))$ as $\{\mathbf{X}_{i,b}(l_{i,b-1}), i\in \SS_{b}(\mathbf{l}_{b-1})\}$, and similarly define $\mathbf{Y}_b(\SS_{b}(\mathbf{l}_{b-1}))$ and $\hat{\mathbf{Y}}_b(\SS_{b}(\mathbf{l}_{b-1}))$.
Then, $(\mathbf{X}_b(\SS_{b}(\mathbf{l}_{b-1})), \hat{\mathbf{Y}}_b(\SS_{b}(\mathbf{l}_{b-1})))$ is independent of $(\mathbf{X}_b(\SS^c_{b}(\mathbf{l}_{b-1})), \hat{\mathbf{Y}}_b(\SS^c_{b}(\mathbf{l}_{b-1})), \mathbf{Y}_b)$, and
$\mbox{Pr}(  \mathcal{A}_{b}(\mathbf{l}_{b-1},\mathbf{l}_b)  )$ can be upper bounded by
\begin{align*}
&2^{T(H(X_{\NN},\hat Y_{\NN}, Y )+\epsilon)}2^{-T(H(X_{\SS^c_{b}(\mathbf{l}_{b-1})},\hat Y_{\SS^c_{b}(\mathbf{l}_{b-1})}, Y )-\epsilon)} 2^{-T(H(X_{\SS_{b}(\mathbf{l}_{b-1})})-\epsilon)}2^{-T(\sum_{i\in \SS_{b}(\mathbf{l}_{b-1})}(H(\hat Y_{i}|X_i)-\epsilon))}\\
=:&2^{-T( \II(\SS_{b}(\mathbf{l}_{b-1}))    -\epsilon'           )}
\end{align*}
where $\II(\SS_{b}(\mathbf{l}_{b-1})) =I(X_{\SS_{b}(\mathbf{l}_{b-1})};\hat{Y}_{\SS^c_{b}(\mathbf{l}_{b-1})},Y|X_{\SS^c_{b}(\mathbf{l}_{b-1})})
-H(\hat{Y}_{\SS_{b}(\mathbf{l}_{b-1})}|X_{\NN},\hat{Y}_{\SS^c_{b}(\mathbf{l}_{b-1})},Y)+\sum_{i \in \SS_{b}(\mathbf{l}_{b-1})}H(\hat Y_i|X_i) $ and $\epsilon'\to 0$ as $\epsilon \to 0$. Then,  we have
\begin{align*}
&\mbox{Pr}(\bigcup_{
\scriptsize \begin{array}{c}
\mathbf{l}_{B+1}^{B+M}: \mathbf{l}_{B+j}\neq \mathbf{1}, \\
\forall j\in[1:M-1]
\end{array}
} \bigcup_{\mathbf{l}_B \neq \mathbf{1}} \bigcap_{b=B+1}^{B+M} \mathcal{A}_{b}(\mathbf{l}_{b-1},\mathbf{l}_b))\\
\leq &  \sum_{
\scriptsize \begin{array}{c}
\mathbf{l}_{B+1}^{B+M}: \mathbf{l}_{B+j}\neq \mathbf{1}, \\
\forall j\in[1:M-1]
\end{array}
} \sum_{\mathbf{l}_B \neq \mathbf{1}} \prod_{b=B+1}^{B+M} \mbox{Pr}(\mathcal{A}_{b}(\mathbf{l}_{b-1},\mathbf{l}_b))  \\
\leq & \sum_{\mathbf{l}_{B+M}} \sum_{\scriptsize \begin{array}{c}
\mathbf{l}_{B+1}^{B+M-1}: \mathbf{l}_{B+j}\neq \mathbf{1}, \\
\forall j\in[1:M-1]
\end{array}} \sum_{\mathbf{l}_B \neq \mathbf{1}}  \prod_{b=B+1}^{B+M} 2^{-T( \II(\SS_{b}(\mathbf{l}_{b-1}))    -\epsilon'           )}  \\
\leq & \sum_{\mathbf{l}_{B+M}} \sum_{
\scriptsize \begin{array}{c}
\SS_{B+1},\ldots,\SS_{B+M}: \\
\SS_{B+j}\neq \emptyset, \forall j\in [1:M]
\end{array}
} \sum_{
\scriptsize \begin{array}{c}
\mathbf{l}_B^{B+M-1}:\SS_{b}(\mathbf{l}_B^{B+M-1})=\SS_{b}, \\
 \forall b\in [B+1:B+M]
\end{array}
}  \prod_{b=B+1}^{B+M} 2^{-T( \II(\SS_{b}(\mathbf{l}_{b-1}))    -\epsilon'           )}  \\
\leq & \sum_{\mathbf{l}_{B+M}}
\sum_{
\scriptsize \begin{array}{c}
\SS_{B+1},\ldots,\SS_{B+M}: \\
\SS_{B+j}\neq \emptyset, \forall j\in [1:M]
\end{array}
}
\prod_{b=B+1}^{B+M} 2^{T( \sum_{i \in \SS_{b}}(I(Y_i;\hat Y_i|X_i)+\epsilon))}
\prod_{b=B+1}^{B+M} 2^{-T( \II(\SS_{b})    -\epsilon'           )}   \\
\leq & \sum_{\mathbf{l}_{B+M}} \sum_{
\scriptsize \begin{array}{c}
\SS_{B+1},\ldots,\SS_{B+M}: \\
\SS_{B+j}\neq \emptyset, \forall j\in [1:M]
\end{array}
}
\prod_{b=B+1}^{B+M} 2^{-T( I(X_{\SS_b}; \hat Y_{\SS_b^c},Y|X_{\SS_b^c}) -  I(Y_{\SS_b};\hat Y_{\SS_b}|X_\NN,Y, \hat Y_{\SS_b^c}) -\epsilon'' )}  \\
\leq & \sum_{\mathbf{l}_{B+M}} \sum_{
\scriptsize \begin{array}{c}
\SS_{B+1},\ldots,\SS_{B+M}: \\
\SS_{B+j}\neq \emptyset, \forall j\in [1:M]
\end{array}
}2^{-T\sum_{b=B+1}^{B+M}( I(X_{\SS_b}; \hat Y_{\SS_b^c},Y|X_{\SS_b^c}) -  I(Y_{\SS_b};\hat Y_{\SS_b}|X_\NN,Y, \hat Y_{\SS_b^c}) -\epsilon'' )}  \\
\leq & \sum_{\mathbf{l}_{B+M}}
(2^n)^M
2^{-TM(\min_{\SS \subseteq \NN: \SS \neq \emptyset}\{ I(X_{\SS}; \hat Y_{\SS^c},Y|X_{\SS^c}) -  I(Y_{\SS};\hat Y_{\SS}|X_\NN,Y, \hat Y_{\SS^c}) -\epsilon''  \})} \\
\leq & 2^{T(\sum_{i\in \NN}(I(\hat Y_i;Y_i|X_i)+\epsilon))}  2^{nM}
2^{-TM(\min_{\SS \subseteq \NN: \SS \neq \emptyset}\{ I(X_{\SS}; \hat Y_{\SS^c},Y|X_{\SS^c}) -  I(Y_{\SS};\hat Y_{\SS}|X_\NN,Y, \hat Y_{\SS^c}) -\epsilon''  \})}
\end{align*}
where $\epsilon'' \to 0$ as $\epsilon \to 0$. Thus, as both $T$ and $M$ go to infinity, the second term in  \dref{E:following} goes to 0, if \dref{cbscondition} holds.

Now consider the first term in \dref{E:following}. For any $j\in [1:M-1]$, we have
\begin{align*}
&\mbox{Pr}(\bigcup_{\mathbf{l}_{B+1}^{B+M}: \mathbf{l}_{B+j}=\mathbf{1}} \bigcup_{\mathbf{l}_B \neq \mathbf{1}} \bigcap_{b=B+1}^{B+M} \mathcal{A}_{b}(\mathbf{l}_{b-1},\mathbf{l}_b))\leq \mbox{Pr}(\bigcup_{\mathbf{l}_{B+1}^{B+j}: \mathbf{l}_{B+j}=\mathbf{1}} \bigcup_{\mathbf{l}_B \neq \mathbf{1}} \bigcap_{b=B+1}^{B+j} \mathcal{A}_{b}(\mathbf{l}_{b-1},\mathbf{l}_b)).
\end{align*}
Note $\mbox{Pr}(\bigcup_{\mathbf{l}_{B+1}^{B+j}: \mathbf{l}_{B+j}=\mathbf{1}} \bigcup_{\mathbf{l}_B \neq \mathbf{1}} \bigcap_{b=B+1}^{B+j} \mathcal{A}_{b}(\mathbf{l}_{b-1},\mathbf{l}_b))$ is the probability that there exists a false $\mathbf{l}_B \neq \mathbf{1}$ satisfies \dref{E:check1} with some $\mathbf{l}_{B+1}^{B+j}$ for any block $b\in [B+1:B+j]$, where $\mathbf{l}_{B+j}=\mathbf{1}$ is true. We can show this probability goes to 0 with the idea of backward decoding as follows.

Specifically, backwards and sequentially from block $b=B+j$ to block $b=B+1$, the destination finds the unique $\mathbf{l}_{b-1}$, such that $(\mathbf{l}_{b-1},\mathbf{l}_{b})$ satisfies \dref{E:check1}, where $\mathbf{l}_{b}$ has already been recovered due to the backwards property of decoding. At each block $b=B+j,B+j-1,\ldots,B+1$, error occurs if the true $\mathbf{l}_{b-1}$ does not satisfy \dref{E:check1}, or a false $\mathbf{l}_{b-1}$ satisfies \dref{E:check1}. According to the properties of typical sequences, the true $\mathbf{l}_{b-1}$ satisfies \dref{E:check1} with high probability.

For a false $\mathbf{l}_{b-1}$ with false $\{l_{i,b-1}, i\in \SS \}$ but true $\{l_{i,b-1}, i\in \SS^c \}$, $(\mathbf{X}_b(\SS), \hat{\mathbf{Y}}_b(\SS))$ is independent of $(\mathbf{X}_b(\SS^c), \hat{\mathbf{Y}}_b(\SS^c), \mathbf{Y}_b)$, and
the probability that $(\mathbf{l}_{b-1},\mathbf{l}_{b})$ satisfies \dref{E:check1} can be upper bounded by
\begin{align*}
2^{T(H(X_{\NN},\hat Y_{\NN}, Y )+\epsilon)}2^{-T(H(X_{\SS^c},\hat Y_{\SS^c}, Y )-\epsilon)} 2^{-T(H(X_{\SS})-\epsilon)}2^{-T(\sum_{i\in \SS}(H(\hat Y_{i}|X_i)-\epsilon))}.
\end{align*}
Since the number of such false $\mathbf{l}_{b-1}$ is upper bounded by $\prod_{i\in \SS}2^{T(I(Y_i;\hat Y_i|X_i) +\epsilon)}$, with the union bound, it is easy to check that the probability of finding such a false $\mathbf{l}_{b-1}$ goes to zero as $T\to \infty$, if \dref{cbscondition} holds.

Therefore, if \dref{cbscondition} holds, the first term in \dref{E:following} also goes to 0 as $T\to \infty$, and  $\mathbf{l}_B$ can be decoded.

b) Given that $\mathbf{l}_B$ has been recovered, the destination performs the backward decoding similar with above. That is, backwards and sequentially from block $b=B$ to block $b=2$, the destination finds the unique $\mathbf{l}_{b-1}$, such that $(\mathbf{l}_{b-1},\mathbf{l}_{b})$ satisfies \dref{E:check1}, where $\mathbf{l}_{b}$ has already been recovered. From the above analysis, it follows that at each block $b=B,B-1,\ldots,2$, the probability of decoding error goes to zero as $T\to \infty$, if \dref{cbscondition} holds. This combined with a) implies that $\mathbf{l}^B$ can be decoded, if \dref{cbscondition} holds.

ii) Then, based on the recovered $\mathbf{l}^B$, the destination finds the unique $\mathbf{m}$ such that for any $b=1,\ldots,B$,
\begin{align}
\left( \mathbf{X}_{b}(m_b), (\mathbf{X}_{1,b}(l_{1,b-1}) , \hat{\mathbf{Y}}_{1,b}(l_{1,b}|l_{1,b-1})),\ldots,(\mathbf{X}_{n,b}(l_{n,b-1}) , \hat{\mathbf{Y}}_{n,b}(l_{n,b}|l_{n,b-1})), \mathbf{Y}_b   \right)\in A_{\epsilon}(X,X_\NN,\hat{Y}_\NN,Y).
\end{align}
Obviously, the probability of decoding error will tend to zero if $R_{\text{C/B/S}}<I(X;\hat{Y}_\NN,Y|X_\NN)$.
\end{proof}

We are now in a position to prove Theorem \ref{T:optallblocks}. To facilitate the proof, we introduce some notations and lemmas.
Let
\begin{align}
J_{\mathcal{A},\BB}(\SS):=&I(X_\SS; \hat Y_{\BB \setminus \SS},\hat Y_{\mathcal{A}}, Y|X_{\mathcal{A}},X_{\BB \setminus \SS}) -  I(Y_{\SS};\hat Y_{\SS}|X_{\mathcal{A}},\hat Y_{\mathcal{A}},Y, X_{\BB}, \hat Y_{\BB\setminus \SS}), \forall \SS \subseteq \BB,\\
J_{\mathcal{B}}(\SS):=&J_{\emptyset,\BB}(\SS)=I(X_\SS; \hat Y_{\BB \setminus \SS},Y|X_{\BB \setminus \SS}) -I(Y_\SS; \hat Y_\SS|X_{\BB}, \hat Y_{\mathcal{B} \setminus \SS},Y), \forall \SS \subseteq \mathcal{B}, \\
J(\SS):=&J_{\NN}(\SS)=I(X_\SS; \hat Y_{\SS^c},Y|X_{\SS^c}) -  I(Y_{\SS};\hat Y_{\SS}|X_\NN,Y, \hat Y_{\SS^c}), \forall \SS \subseteq \NN.
\end{align}
Then, we have the following lemmas, whose proofs will be presented in Appendix \ref{A:B}.
\begin{lemma}
\label{L:irrunion}
1) If $J_{\mathcal{A}}(\SS_1)\geq 0$, $\forall \SS_1 \subseteq \mathcal{A}$, and $J_{\mathcal{B}}(\SS_2)\geq 0$, $\forall \SS_2 \subseteq \mathcal{B}$, then $J_{\mathcal{A}\cup \BB}(\SS)\geq 0$, $\forall \SS \subseteq \mathcal{A}\cup \BB.$

2) If $J_{\mathcal{A}}(\SS_1)\geq 0$, $\forall \SS_1 \subseteq \mathcal{A}$, and $J_{\mathcal{A},\mathcal{B}}(\SS_2)\geq 0$, $\forall \SS_2 \subseteq \mathcal{B}$, then $J_{\mathcal{A}\cup \BB}(\SS)\geq 0$, $\forall \SS \subseteq \mathcal{A}\cup \BB.$
\end{lemma}

\begin{lemma}
\label{L:irrlargestD}
Under any $p(x)\prod_{i=1}^{n}p(x_i)p(\hat y_i|x_i,y_i)$, there exists a unique set $\DD$, which is the largest subset of $\NN$ satisfying
$$J_{\DD}(\SS)\geq 0, \forall \SS \subseteq \DD.$$
\end{lemma}

\begin{lemma}
\label{L:irrexsit}
If $J_{\mathcal{A},\BB}(\BB)\geq 0$ for some nonempty $\BB$, then there exists some nonempty $\mathcal{C} \subseteq \BB$ such that $J_{\mathcal{A},\CC}(\SS)\geq 0, \forall \SS \subseteq \mathcal{C}$.
\end{lemma}

\begin{lemma}
\label{L:irraub}
For any $\mathcal{A}$ and $\BB$ with $\mathcal{A} \cap \BB=\emptyset$, $J(\mathcal{A})+J(\mathcal{B})=J(\mathcal{A} \cup \BB)+J(\mathcal{A}\circ \mathcal{B}),$
where
\begin{align*}
J(\mathcal{A}\circ \mathcal{B})=&I(X_{\mathcal{A}},\hat Y_{\mathcal{A}};X_{\mathcal{B}},\hat Y_{\mathcal{B}}|X_{(\mathcal{A}\cup \mathcal{B})^c}, \hat Y_{(\mathcal{A}\cup \mathcal{B})^c},Y) \\
=&I(X_{\mathcal{A}};X_{\mathcal{B}}|X_{(\mathcal{A}\cup \mathcal{B})^c}, \hat Y_{(\mathcal{A}\cup \mathcal{B})^c},Y)+I(X_{\mathcal{A}};\hat Y_{\mathcal{B}}|X_{\mathcal{A}^c},\hat Y_{(\mathcal{A}\cup \mathcal{B})^c},Y)\\
&+I(X_{\mathcal{B}};\hat Y_{\mathcal{A}}|X_{\mathcal{B}^c},\hat Y_{(\mathcal{A}\cup \mathcal{B})^c},Y)+I(\hat Y_{\mathcal{A}};\hat Y_{\mathcal{B}}|X_{\mathcal{N}},\hat Y_{(\mathcal{A}\cup \mathcal{B})^c},Y).
\end{align*}
\end{lemma}

The proof of Theorem \ref{T:optallblocks} is similar to the proof of Theorem \ref{optcf}, and the details are as follows.

\begin{proof}[Proof of Theorem \ref{T:optallblocks}]
$R_{\text{C/B/S}}^*$ and $R_{\text{R/U/J}}^*$ can be respectively written as
\begin{align}
R_{\text{C/B/S}}^*= &\max_{p(x)\prod_{i=1}^{n}p(x_i)p(\hat y_i|x_i,y_i)} I(X;\hat Y_\NN, Y|X_\NN) \label{cbsmax} \\
\text{such that~} &J(\SS)\geq 0 , \forall \SS \subseteq \NN,\label{cbsmaxconstraint}
\end{align}
and
\begin{align}
R_{\text{R/U/J}}^*=&\max_{p(x)\prod_{i=1}^{n}p(x_i)p(\hat y_i|x_i,y_i)}\min_{\SS \subseteq \NN}I(X,X_\SS;\hat Y_{\SS^c},Y|X_{\SS^c})-I(Y_{\SS};\hat Y_{\SS}|X,X_\NN,Y, \hat Y_{\SS^c})\nonumber \\
=&\max_{p(x)\prod_{i=1}^{n}p(x_i)p(\hat y_i|x_i,y_i)}\min_{\SS \subseteq \NN}\{I(X;\hat Y_\NN, Y|X_\NN)+J(\SS)   \}.\label{rujmax}
\end{align}

We show  $R_{\text{C/B/S}}^* = R_{\text{R/U/J}}^*$ by showing that  $R_{\text{C/B/S}}^*\leq  R_{\text{R/U/J}}^*$ and  $R_{\text{C/B/S}}^* \geq  R_{\text{R/U/J}}^*$ respectively. Under any $p(x)\prod_{i=1}^{n}p(x_i)p(\hat y_i|x_i,y_i)$ such that $J(\SS)\geq 0$, $\forall \SS \subseteq \NN$, we have
$$\min_{\SS \subseteq \NN}\{I(X;\hat Y_\NN, Y|X_\NN)+J(\SS)   \}=I(X;\hat Y_\NN, Y|X_\NN),$$
and thus $R_{\text{C/B/S}}^*\leq  R_{\text{R/U/J}}^*$.

To show $R_{\text{C/B/S}}^* \geq R_{\text{R/U/J}}^*$, it is sufficient to show that $R_{\text{R/U/J}}^*$ can be achieved only with the distribution $p(x)\prod_{i=1}^{n}p(x_i)p(\hat y_i|x_i,y_i)$ such that $J(\SS)\geq 0$, $\forall \SS \subseteq \NN$. We will show this by two steps as follows: i) We first show that under any $p(x)\prod_{i=1}^{n}p(x_i)p(\hat y_i|x_i,y_i)$, if $\DD^c \neq \emptyset$, then $\DD^c \in \argmin_{\SS \subseteq \NN}J(\SS)$ and $\bigcap_{\TT \in \argmin_{\SS \subseteq \NN}J(\SS)}\TT=\DD^c$, where $\DD$ is defined as in Lemma \ref{L:irrlargestD} and $\argmin_{\SS \subseteq \NN}J(\SS):=\{\TT \subseteq \NN: J(\TT)=\min_{\SS \subseteq \NN}J(\SS)\}$. ii) We then argue that,  under the optimal $p(x)\prod_{i=1}^{n}p(x_i)p(\hat y_i|x_i,y_i)$, $\DD^c$ must be $\emptyset$, i.e., $\DD$ must be $\NN$, and thus by the definition of $\DD$, $J(\SS)\geq0, \forall \SS \subseteq \NN$.

i) Assuming $\DD^c \neq \emptyset$ throughout Part i), we show $\DD^c \in \argmin_{\SS \subseteq \NN}J(\SS)$ and $\bigcap_{\TT \in \argmin_{\SS \subseteq \NN}J(\SS)}\TT=\DD^c$.

1) We first show $J(\DD^c)<0$ by using a contradiction argument. Suppose $J(\DD^c) \geq 0$, i.e., $J_{\DD,\DD^c}(\DD^c) \geq 0$. Then, by  Lemma \ref{L:irrexsit}, we have that there exists some nonempty $\BB \subseteq \DD^c$ such that $J_{\DD,\BB}(\SS) \geq 0$, $\forall \SS \subseteq \BB$. This will further imply, by Part 2) of Lemma \ref{L:irrunion}, that $J_{\DD \cup \BB}(\SS)\geq 0, \forall \SS \subseteq \DD \cup \BB$. This is contradictory with the definition of $\DD$, and thus $J(\DD^c)<0$.

2) We show that $\forall \mathcal{A} \subseteq \DD^c$ and $\mathcal{A} \neq \DD^c$, $J(\mathcal{A})>J(\DD^c)$, and thus $J(\mathcal{A})>\min_{\SS \subseteq \NN}J(\mathcal{S})$. The proof is still by contradiction. Suppose that there exists some $\mathcal{A} \subseteq \DD^c$ and $\mathcal{A} \neq \DD^c$ such that $J(\mathcal{A}) \leq J(\DD^c)$. Then $J(\DD^c)-J(\mathcal{A})\geq 0$, i.e.,
\begin{align*}
&I(X_{\DD^c}; \hat Y_{\DD},Y|X_{\DD}) -  I(Y_{\DD^c};\hat Y_{\DD^c}|X_\NN,Y, \hat Y_{\DD})-I(X_\mathcal{A}; \hat Y_{\mathcal{A}^c},Y|X_{\mathcal{A}^c}) + I(Y_{\mathcal{A}};\hat Y_{\mathcal{A}}|X_\NN,Y, \hat Y_{\mathcal{A}^c})\\
=&I(X_{\DD^c \setminus \mathcal{A}}; \hat Y_{\DD},Y|X_{\DD})+I(X_{\mathcal{A}}; \hat Y_{\DD},Y|X_{ \mathcal{A}^c})
- I(Y_{\DD^c \setminus \mathcal{A}};\hat Y_{\DD^c \setminus \mathcal{A} }|X_\NN,Y, \hat Y_{\DD})-I(Y_{\mathcal{A}};\hat Y_{\mathcal{A} }|X_\NN,Y, \hat Y_{\mathcal{A}^c})\\
&-I(X_{\mathcal{A}}; \hat Y_{\DD},Y|X_{ \mathcal{A}^c})-I(X_{\mathcal{A}}; \hat Y_{\DD^c \setminus \mathcal{A}}|\hat Y_{\DD},Y,X_{ \mathcal{A}^c})+
I(Y_{\mathcal{A}};\hat Y_{\mathcal{A} }|X_\NN,Y, \hat Y_{\mathcal{A}^c})\\
=&I(X_{\DD^c \setminus \mathcal{A}}; \hat Y_{\DD},Y|X_{\DD})-H(\hat Y_{\DD^c \setminus \mathcal{A} }|X_\NN,Y, \hat Y_{\DD})+H(\hat Y_{\DD^c \setminus \mathcal{A} }|Y_{\DD^c \setminus \mathcal{A}},X_\NN,Y, \hat Y_{\DD})\\
&-H(\hat Y_{\DD^c \setminus \mathcal{A}}|\hat Y_{\DD},Y,X_{ \mathcal{A}^c})+H(\hat Y_{\DD^c \setminus \mathcal{A}}|X_{\mathcal{A}},\hat Y_{\DD},Y,X_{ \mathcal{A}^c})\\
=&I(X_{\DD^c \setminus \mathcal{A}}; \hat Y_{\DD},Y|X_{\DD})-I(Y_{\DD^c \setminus \mathcal{A}};\hat Y_{\DD^c \setminus \mathcal{A} }|X_\DD,X_{\DD^c \setminus \mathcal{A}} ,Y, \hat Y_{\DD})\\
=&J_{\DD,\DD^c \setminus \mathcal{A}}(\DD^c \setminus \mathcal{A})\\
\geq & 0.
\end{align*}
Again by Lemma \ref{L:irrexsit} and \ref{L:irrunion} successively, we can conclude that there exists some nonempty $\BB \subseteq \DD^c\setminus \mathcal{A}$, such that $J_{\DD \cup \BB}(\SS)\geq 0, \forall \SS \subseteq \DD \cup \BB$, which is in contradiction. Therefore, $J(\mathcal{A})>J(\DD^c)\geq \min_{\SS \subseteq \NN}J(\mathcal{S})$.

3) We prove that $\forall \mathcal{A}$ with $\mathcal{A}\DD \neq \emptyset$ and $\mathcal{A}\DD^c \neq \DD^c$, $J(\mathcal{A})>J(\mathcal{A} \cup \DD^c) \geq \min_{\SS \subseteq \NN}J(\mathcal{S})$. Let $\mathcal{A}_1=\mathcal{A}\DD$ and $\mathcal{A}_2=\mathcal{A}\DD^c$. Then, we have, by Lemma \ref{L:irraub}, that
\begin{align*}
J(\mathcal{A})=&J(\mathcal{A}_1 \cup \mathcal{A}_2)=J(\mathcal{A}_1)+J(\mathcal{A}_2)-J(\mathcal{A}_1\circ \mathcal{A}_2),\\
J(\mathcal{A}_1 \cup \DD^c)=&J(\mathcal{A}_1)+J(\DD^c)-J(\mathcal{A}_1 \circ \DD^c).
\end{align*}
Since $J(\mathcal{A}_2)>J(\DD^c)$ by 2), to show $J(\mathcal{A})>J(\mathcal{A} \cup \DD^c) \geq \min_{\SS \subseteq \NN}J(\mathcal{S})$, we only need to show $J(\mathcal{A}_1\circ \mathcal{A}_2)\leq J(\mathcal{A}_1 \circ \DD^c)$. Let $\mathcal{A}_3=\DD^c \setminus \mathcal{A}_2$. Then, we have
\begin{align*}
&J(\mathcal{A}_1 \circ \DD^c)- J(\mathcal{A}_1\circ \mathcal{A}_2)\\
=&I(X_{\mathcal{A}_1};X_{\mathcal{A}_2 \cup \mathcal{A}_3}|X_{(\mathcal{A}_1 \cup \mathcal{A}_2 \cup \mathcal{A}_3)^c}, \hat Y_{(\mathcal{A}_1 \cup \mathcal{A}_2 \cup \mathcal{A}_3)^c},Y)+I(X_{\mathcal{A}_1};\hat Y_{\mathcal{A}_2 \cup \mathcal{A}_3}|X_{\mathcal{A}_1^c},\hat Y_{(\mathcal{A}_1 \cup \mathcal{A}_2 \cup \mathcal{A}_3)^c},Y)\\
&+I(X_{\mathcal{A}_2 \cup \mathcal{A}_3};\hat Y_{\mathcal{A}_1}|X_{(\mathcal{A}_2 \cup \mathcal{A}_3)^c},\hat Y_{(\mathcal{A}_1 \cup \mathcal{A}_2 \cup \mathcal{A}_3)^c},Y)+I(\hat Y_{\mathcal{A}_1};\hat Y_{\mathcal{A}_2 \cup \mathcal{A}_3}|X_{\mathcal{N}},\hat Y_{(\mathcal{A}_1 \cup \mathcal{A}_2 \cup \mathcal{A}_3)^c},Y)\\
&-I(X_{\mathcal{A}_1};X_{\mathcal{A}_2 }|X_{(\mathcal{A}_1 \cup \mathcal{A}_2 )^c}, \hat Y_{(\mathcal{A}_1 \cup \mathcal{A}_2 )^c},Y)-I(X_{\mathcal{A}_1};\hat Y_{\mathcal{A}_2 }|X_{\mathcal{A}_1^c},\hat Y_{(\mathcal{A}_1 \cup \mathcal{A}_2 )^c},Y)\\
&-I(X_{\mathcal{A}_2};\hat Y_{\mathcal{A}_1}|X_{\mathcal{A}_2 ^c},\hat Y_{(\mathcal{A}_1 \cup \mathcal{A}_2 )^c},Y)-I(\hat Y_{\mathcal{A}_1};\hat Y_{\mathcal{A}_2 }|X_{\mathcal{N}},\hat Y_{(\mathcal{A}_1 \cup \mathcal{A}_2)^c},Y)\\
=&I(X_{\mathcal{A}_1};X_{ \mathcal{A}_3}|X_{(\mathcal{A}_1 \cup \mathcal{A}_2 \cup \mathcal{A}_3)^c}, \hat Y_{(\mathcal{A}_1 \cup \mathcal{A}_2 \cup \mathcal{A}_3)^c},Y)+I(X_{\mathcal{A}_1};X_{\mathcal{A}_2},\hat Y_{\mathcal{A}_3}|X_{(\mathcal{A}_1 \cup \mathcal{A}_2)^c},\hat Y_{(\mathcal{A}_1 \cup \mathcal{A}_2 \cup \mathcal{A}_3)^c},Y)\\
&+I(X_{\mathcal{A}_3};\hat Y_{\mathcal{A}_1}|X_{(\mathcal{A}_2 \cup \mathcal{A}_3)^c},\hat Y_{(\mathcal{A}_1 \cup \mathcal{A}_2 \cup \mathcal{A}_3)^c},Y)+I(\hat Y_{\mathcal{A}_1};X_{\mathcal{A}_2}, \hat Y_{\mathcal{A}_3}|X_{\mathcal{A}_2^c},\hat Y_{(\mathcal{A}_1 \cup \mathcal{A}_2 \cup \mathcal{A}_3)^c},Y)\\
&-I(X_{\mathcal{A}_1};X_{\mathcal{A}_2 }|X_{(\mathcal{A}_1 \cup \mathcal{A}_2 )^c}, \hat Y_{(\mathcal{A}_1 \cup \mathcal{A}_2 )^c},Y)-I(X_{\mathcal{A}_2};\hat Y_{\mathcal{A}_1}|X_{\mathcal{A}_2 ^c},\hat Y_{(\mathcal{A}_1 \cup \mathcal{A}_2 )^c},Y)\\
=&I(X_{\mathcal{A}_1};X_{ \mathcal{A}_3}|X_{(\mathcal{A}_1 \cup \mathcal{A}_2 \cup \mathcal{A}_3)^c}, \hat Y_{(\mathcal{A}_1 \cup \mathcal{A}_2 \cup \mathcal{A}_3)^c},Y)+I(X_{\mathcal{A}_1};\hat Y_{\mathcal{A}_3}|X_{(\mathcal{A}_1 \cup \mathcal{A}_2)^c},\hat Y_{(\mathcal{A}_1 \cup \mathcal{A}_2 \cup \mathcal{A}_3)^c},Y)\\
&+I(X_{\mathcal{A}_3};\hat Y_{\mathcal{A}_1}|X_{(\mathcal{A}_2 \cup \mathcal{A}_3)^c},\hat Y_{(\mathcal{A}_1 \cup \mathcal{A}_2 \cup \mathcal{A}_3)^c},Y)+I(\hat Y_{\mathcal{A}_1}; \hat Y_{\mathcal{A}_3}|X_{\mathcal{A}_2^c},\hat Y_{(\mathcal{A}_1 \cup \mathcal{A}_2 \cup \mathcal{A}_3)^c},Y)\\
\geq &0.
\end{align*}
Thus, we have $J(\mathcal{A})>J(\mathcal{A}_1 \cup \DD^c)\geq \min_{\SS \subseteq \NN}J(\mathcal{S})$.

4) We prove that $\forall \mathcal{A}$ with $\mathcal{A}\DD \neq \emptyset$ and $\mathcal{A}\DD^c = \DD^c$, $J(\mathcal{A})\geq J(\DD^c)$. Letting $\mathcal{A}_1=\mathcal{A}\DD$, we have
\begin{align*}
J(\mathcal{A})=J(\mathcal{A}_1 \cup \DD^c)=J(\mathcal{A}_1)+J( \DD^c)-J(\mathcal{A}_1\circ  \DD^c).
\end{align*}
Thus, to show $J(\mathcal{A})\geq J(\DD^c)$, we only need to show $J(\mathcal{A}_1)-J(\mathcal{A}_1\circ  \DD^c)\geq 0$. For this, we have
\begin{align*}
&J(\mathcal{A}_1)-J(\mathcal{A}_1\circ  \DD^c)\\
=& I(X_{\mathcal{A}_1}; \hat Y_{\DD^c},\hat Y_{\DD \setminus \mathcal{A}_1},Y|X_{\DD^c},X_{\DD \setminus \mathcal{A}_1}) -  I(Y_{\mathcal{A}_1};\hat Y_{\mathcal{A}_1}|X_\NN,Y, \hat Y_{\DD^c},\hat Y_{\DD \setminus \mathcal{A}_1})\\
&-I(X_{\mathcal{A}_1},\hat Y_{\mathcal{A}_1}; X_{\DD^c},\hat Y_{\DD^c}|X_{\DD \setminus \mathcal{A}_1},\hat Y_{\DD \setminus \mathcal{A}_1},Y)\\
=&I(X_{\mathcal{A}_1}; X_{\DD^c},\hat Y_{\DD^c},\hat Y_{\DD \setminus \mathcal{A}_1},Y|X_{\DD \setminus \mathcal{A}_1}) -  I(Y_{\mathcal{A}_1};\hat Y_{\mathcal{A}_1}|X_\NN,Y, \hat Y_{\DD^c},\hat Y_{\DD \setminus \mathcal{A}_1})\\
&-I(X_{\mathcal{A}_1}; X_{\DD^c},\hat Y_{\DD^c}|X_{\DD \setminus \mathcal{A}_1},\hat Y_{\DD \setminus \mathcal{A}_1},Y)-I(\hat Y_{\mathcal{A}_1}; X_{\DD^c},\hat Y_{\DD^c}|X_{\DD},\hat Y_{\DD \setminus \mathcal{A}_1},Y)\\
=&I(X_{\mathcal{A}_1}; \hat Y_{\DD \setminus \mathcal{A}_1},Y|X_{\DD \setminus \mathcal{A}_1})-I(\hat Y_{\mathcal{A}_1}; X_{\DD^c},\hat Y_{\DD^c},Y_{\mathcal{A}_1}|X_{\DD},\hat Y_{\DD \setminus \mathcal{A}_1},Y)\\
=&J_{\DD}(\mathcal{A}_1)\\
\geq & 0,
\end{align*}
and thus $J(\mathcal{A})\geq J(\DD^c)$.

Combining 2)-4), we can conclude that $\DD^c \in \argmin_{\SS \subseteq \NN}J(\SS)$ and $\bigcap_{\TT \in \argmin_{\SS \subseteq \NN}J(\SS)}\TT=\DD^c$.

ii) We now argue that under the optimal $p(x)\prod_{i=1}^{n}p(x_i)p(\hat y_i|x_i,y_i)$ that achieves $R^*_{\text{R/U/J}}$, if $\DD^c \neq \emptyset$, then $R^*_{\text{R/U/J}}$ is not optimal; and hence $\DD^c$ must be $\emptyset$.

Suppose $\DD^c \neq \emptyset$ at the optimum. Then, $\DD^c \in \argmin_{\SS \subseteq \NN}J(\SS)$ and $\bigcap_{\TT \in \argmin_{\SS \subseteq \NN}J(\SS)}\TT=\DD^c$. Therefore,
\begin{align}
R^*_{\text{ R/U/J}}=&I(X,X_{\DD^c};\hat Y_{\DD},Y|X_{\DD})-I(Y_{\DD^c};\hat Y_{\DD^c}|X,X_\NN,Y, \hat Y_{\DD})   \label{E:seen1}\\
=& I(X,X_\TT;\hat Y_{\TT^c},Y|X_{\TT^c})-I(Y_{\TT};\hat Y_{\TT}|X,X_\NN,Y, \hat Y_{\TT^c}) ,\label{E:seen2}
\end{align}
for any $\TT \in \argmin_{\SS \subseteq \NN}J(\SS)$, $\TT \neq \DD^c$.

We argue that higher rate can be achieved. Consider $\hat Y'_1,\hat Y'_2,\ldots,\hat Y'_n$, where $\hat Y'_i=\hat Y_i$ for any $i\in \DD$, and $\hat Y'_i=\hat Y_i$ with probability $p$ and $\hat Y'_i=\emptyset$ with probability $1-p$ for any $i\in \DD^c$. When $p=1$, the achievable rate with $\hat Y'_1,\hat Y'_2,\ldots,\hat Y'_n$ is $R^*_{\text{R/U/J}}$. As $p$ decreases from 1, in \dref{E:seen1} and \dref{E:seen2},  both $$I(X,X_{\DD^c};\hat Y_{\DD},Y|X_{\DD})-I(Y_{\DD^c};\hat Y_{\DD^c}|X,X_\NN,Y, \hat Y_{\DD})$$ and $$I(X,X_\TT;\hat Y_{\TT^c},Y|X_{\TT^c})-I(Y_{\TT};\hat Y_{\TT}|X,X_\NN,Y, \hat Y_{\TT^c})$$ will increase, where $\TT \in \argmin_{\SS \subseteq \NN}J(\SS)$, $\TT \neq \DD^c$.  Thus, no matter how
$$I(X,X_\SS;\hat Y_{\SS^c},Y|X_{\SS^c})-I(Y_{\SS};\hat Y_{\SS}|X,X_\NN,Y, \hat Y_{\SS^c})$$  will change as $p$ decreases for $\SS \notin \argmin_{\SS \subseteq \NN}J(\SS)$, it is certain that there exists a $p^*$ such that the achievable rate by using $\hat Y'_1,\hat Y'_2,\ldots,\hat Y'_n$ is larger than $R^*_{\text{R/U/J}}$. This is in contradiction with the optimality of $R^*_{\text{R/U/J}}$, and thus at the optimum, $\DD^c$ must be $\emptyset$ , i.e., $J(\SS)\geq 0$, $\forall \SS \subseteq \NN$. This completes the proof of Theorem \ref{T:optallblocks}.
\end{proof}

\subsection{Cumulative encoding/block-by-block backward decoding/compression-message joint decoding and Necessity of joint decodability}

Some notations and lemmas are introduced to facilitate the later discussion.   Let
\begin{align*}
K_{\mathcal{A},\BB}(\SS):=&I(X_\SS; \hat Y_{\BB \setminus \SS},\hat Y_{\mathcal{A}}, Y|X,X_{\mathcal{A}},X_{\BB \setminus \SS}) -  I(Y_{\SS};\hat Y_{\SS}|X,X_{\mathcal{A}},\hat Y_{\mathcal{A}},Y, X_{\BB}, \hat Y_{\BB\setminus \SS}), \forall \SS \subseteq \BB,\\
K_{\mathcal{B}}(\SS):=&K_{\emptyset,\BB}(\SS)=I(X_\SS; \hat Y_{\BB \setminus \SS},Y|X,X_{\BB \setminus \SS}) -I(Y_\SS; \hat Y_\SS|X,X_{\BB}, \hat Y_{\mathcal{B} \setminus \SS},Y), \forall \SS \subseteq \mathcal{B}, \\
R_{\mathcal{B}}(\SS):=&
I(X,X_\SS; \hat Y_{\BB \setminus \SS},Y|X_{\BB \setminus \SS}) -I(Y_\SS; \hat Y_\SS|X,X_{\BB}, \hat Y_{\mathcal{B} \setminus \SS},Y) , \forall \SS \subseteq \mathcal{B}.
\end{align*}

\begin{lemma}
\label{union}
1) If $K_{\mathcal{A}}(\SS_1)> 0$, for any nonempty $\SS_1 \subseteq \mathcal{A}$, and $K_{\mathcal{B}}(\SS_2)> 0$, for any nonempty $ \SS_2 \subseteq \mathcal{B}$, then $K_{\mathcal{A}\cup \BB}(\SS)> 0$, for any nonempty $ \SS \subseteq \mathcal{A}\cup \BB.$

2) If $K_{\mathcal{A}}(\SS_1)> 0$, for any nonempty $\SS_1 \subseteq \mathcal{A}$, and $K_{\mathcal{A},\mathcal{B}}(\SS_2)> 0$, for any nonempty $ \SS_2 \subseteq \mathcal{B}$, then $K_{\mathcal{A}\cup \BB}(\SS)> 0$, for any nonempty $ \SS \subseteq \mathcal{A}\cup \BB.$
\end{lemma}

\begin{lemma}
\label{L:largest}
Under any $p(x)\prod_{i=1}^{n}p(x_i)p(\hat y_i|x_i,y_i)$, there exists a unique set $\DD_{\text{J}}$, which is the largest subset of $\NN$ satisfying
$$K_{\DD_{\text{J}}}(\SS)>0, \forall \SS \subseteq\DD_{\text{J}}, \SS \neq \emptyset.$$
\end{lemma}

\begin{lemma}
\label{exsit}
If $K_{\mathcal{A},\BB}(\BB)> 0$ for some nonempty $\BB$, then there exists some nonempty $\mathcal{C} \subseteq \BB$ such that $K_{\mathcal{A},\CC}(\SS)> 0$, for any nonempty $ \SS \subseteq \mathcal{C}.$
\end{lemma}

\begin{lemma}
\label{L:Runion}
For any disjoint $\mathcal{A}$ and $\mathcal{B}$, and any $\SS \subseteq \mathcal{A}\cup \BB$, let $\SS_1=\mathcal{S}\mathcal{A}$ and $\SS_2=\SS \BB $. Then, we have:

1) $R_{\mathcal{A}\cup \BB}(\SS)\geq R_{\mathcal{A}}(\SS_1)+K_{\mathcal{A}\cup \BB}(\SS_2)$.

2) Specially, when $\SS_2=\BB  $, $R_{\mathcal{A}\cup \BB}(\SS )= R_{\mathcal{A}}(\SS_1)+K_{\mathcal{A},\BB  }(\BB  ).$
\end{lemma}

Lemmas \ref{union}-\ref{exsit} can be proved along the same lines as the proofs of Lemmas \ref{L:irrunion}-\ref{L:irrexsit} respectively, while the proof of Lemma \ref{L:Runion} is given in Appendix \ref{A:C}.

The cumulative encoding/block-by-block backward decoding/compression-message joint decoding scheme is presented in the following proof.

\begin{proof}[ Proof of Theorem \ref{T:cbj}] The uniqueness of  $\DD_{\text{J}}$ has been established in Lemma \ref{L:largest}. Below, we focus on showing that i) the rate in  \dref{cbjrate} is achievable, and ii) the compressions in the set $\DD_{\text{J}}$ can be decoded jointly with $X$.
%
%

To make the presentation easier to follow, we first consider the case when $\DD_{\text{J}} = \NN$, i.e., the case when
\begin{align}I(X_\SS; \hat Y_{\SS^c},Y|X,X_{\SS^c}) -  I(Y_{\SS};\hat Y_{\SS}|X,X_\NN,Y, \hat Y_{\SS^c})>0, \forall \SS \subseteq \NN, \SS \neq \emptyset,\label{E:assume1}
\end{align}
and show that
\begin{align}R_{\text{C/B/J}}<\min_{\SS \subseteq \mathcal{N}}I(X,X_\SS;\hat Y_{\SS^c},Y|X_{\SS^c})-I(Y_{\SS};\hat Y_{\SS}|X,X_\NN,Y, \hat Y_{\SS^c})\label{E:assume2}\end{align}
is achievable. The case of $\DD_{\text{J}} \neq \NN$ will follow immediately after the case of $\DD_{\text{J}} = \NN$ is treated.

Fix $p(x)\prod_{i=1}^{n}p(x_i)p(\hat y_i|x_i,y_i)$. Assume \dref{E:assume1} holds. The codebook generation and encoding process here are exactly the same as those in the proof of Theorem \ref{T:cbs}, and hence omitted. For the decoding, the destination finds the unique message vector $\mathbf{m}=(m_1,m_2,\ldots,m_B)$ and some $\mathbf{l}^{B+M}=(\mathbf{l}_1,\ldots,\mathbf{l}_{B+M})$ such that for any $b=1,\ldots,B+M$,
\begin{align}
\left( \mathbf{X}_{b}(m_b), (\mathbf{X}_{1,b}(l_{1,b-1}) , \hat{\mathbf{Y}}_{1,b}(l_{1,b}|l_{1,b-1})),\ldots,(\mathbf{X}_{n,b}(l_{n,b-1}) , \hat{\mathbf{Y}}_{n,b}(l_{n,b}|l_{n,b-1})), \mathbf{Y}_b   \right)\in A_{\epsilon}(X,X_\NN,\hat{Y}_\NN,Y),\label{E:chk}
\end{align}
where $m_b=1$ is dummy message for all $b\in [B+1:B+M]$.

Again, this can be done backwardly as follows.

a) The destination first finds the unique $\mathbf{l}_B$ such that there exists some $\mathbf{l}_{B+1}^{B+M}=(\mathbf{l}_{B+1},\ldots,\mathbf{l}_{B+M})$ satisfying \dref{E:chk}
for any $b=B+1,\ldots,B+M$. Through the similar lines as in the proof of Theorem \ref{T:cbs} with $\mathbf{X}_{b}(m_b), b\in [B+1:B+M]$ taken into account and treated as known signals, it follows that  $\mathbf{l}_B$ can be decoded if \dref{E:assume1} holds.

b) Backwards and sequentially from block $b=B$ to block $b=1$, the destination finds the unique pair $(m_b,\mathbf{l}_{b-1})$, such that $(m_b, \mathbf{l}_{b-1})$ satisfies \dref{E:chk}, where $\mathbf{l}_{b}$ has already been recovered due to the backwards property of decoding.

At each block $b=B,B-1,\ldots,1$, error occurs with $m_b$ if the true $m_b$ does not satisfy \dref{E:chk} with any $\mathbf{l}_{b-1}$, or a false $m_b$ satisfies \dref{E:chk} with some $\mathbf{l}_{b-1}$. According to the properties of typical sequences, the true $(m_b,\mathbf{l}_{b-1})$ satisfies \dref{E:chk} with high probability.

For a false $m_b$ and a $\mathbf{l}_{b-1}$ with false $\{l_{i,b-1}, i\in \SS \}$ but true $\{l_{i,b-1}, i\in \SS^c \}$, $\mathbf{X}_b(m_b)$ and $(\mathbf{X}_b(\SS), \hat{\mathbf{Y}}_b(\SS))$ and $(\mathbf{X}_b(\SS^c), \hat{\mathbf{Y}}_b(\SS^c), \mathbf{Y}_b)$ are mutually independent, and
the probability that $(m_b, \mathbf{l}_{b-1})$ satisfies \dref{E:chk} can be upper bounded by
\begin{align*}
2^{T(H(X,X_{\NN},\hat Y_{\NN}, Y )+\epsilon)}2^{-T(H(X )-\epsilon)} 2^{-T(H(X_{\SS^c},\hat Y_{\SS^c}, Y )-\epsilon)} 2^{-T(H(X_{\SS})-\epsilon)}2^{-T(\sum_{i\in \SS}(H(\hat Y_{i}|X_i)-\epsilon))}.
\end{align*}
Since the number of such false $(m_b,\mathbf{l}_{b-1})$ is upper bounded by $2^{TR_{\text{C/B/J}}}\prod_{i\in \SS}2^{T(I(Y_i;\hat Y_i|X_i) +\epsilon)}$, with the union bound, it is easy to check that the probability of finding a false $m_b$ goes to zero as $T\to \infty$, if \dref{E:assume2} holds.

Then, based on the recovered $m_b$ and $\mathbf{l}_{b}$, again from the proof of Theorem \ref{T:cbs} with $\mathbf{X}_{b}(m_b)$ taken into account and treated as known signal, it follows that $\mathbf{l}_{b-1}$ can be decoded if \dref{E:assume1} holds.

Combining a) and b), we can conclude that both $\mathbf{m}$ and $\mathbf{l}^B$ can be decoded if both \dref{E:assume1} and \dref{E:assume2} hold.

If under $p(x)\prod_{i=1}^{n}p(x_i)p(\hat y_i|x_i,y_i)$, $\DD_{\text{J}} \neq \NN$, then through the same line as above with $\NN$ replaced by $\DD_{\text{J}}$, it readily follows that
$$R_{\text{C/B/J}} < \min_{\SS \subseteq \mathcal{D}_{\text{J}}} I(X,X_\SS;\hat Y_{\mathcal{D}_{\text{J}} \setminus \SS},Y|X_{\mathcal{D}_{\text{J}} \setminus \SS})-I(Y_{\SS};\hat Y_{\SS}|X,X_{{\mathcal{D}_{\text{J}}}},Y, \hat Y_{\mathcal{D}_{\text{J}} \setminus \SS})$$
is achievable; and $\hat Y_{\DD_{\text{J}}}$, or more strictly, $\{l_{i}^{B}, i\in \DD_{\text{J}} \}$, can be decoded jointly with $X$ since
$$I(X_\SS;\hat Y_{\DD_{\text{J}} \setminus \SS},Y|X,X_{\DD_{\text{J}} \setminus \SS})-I(Y_{\SS};\hat Y_{\SS}|X,X_{\DD_{\text{J}}},Y, \hat Y_{\DD_{\text{J}} \setminus \SS})> 0,$$
for any nonempty $\SS \subseteq \DD_{\text{J}}$.
\end{proof}

Now, we demonstrate that only those relay nodes, whose compressions can be eventually decoded, are helpful to the decoding of the original message.

\begin{proof}
[Proof of Theorem \ref{T:equi}]
The uniqueness of  $\DD_{\text{J}}$ has been treated in Lemma \ref{L:largest}, while the uniqueness of $\DD'_{\text{J}}$ can be established along the same lines. To prove Theorem \ref{T:equi}, in terms of the notations defined in this section, we will sequentially prove that: i) $\max_{\MM \subseteq \NN} \min_{\SS \subseteq \MM}R_{\MM}(\SS )= \min_{\SS \subseteq \DD_{\text{J}}}R_{\DD_{\text{J}}}(\SS )$; ii) $ \min_{\SS \subseteq \MM}R_{\MM}(\SS )< \min_{\SS \subseteq \DD'_{\text{J}}}R_{\DD'_{\text{J}}}(\SS )$, for any $\MM \nsubseteq \DD'_{\text{J}}$; iii)  $\max_{\MM \subseteq \NN} \min_{\SS \subseteq \MM}R_{\MM}(\SS )= \min_{\SS \subseteq \DD'_{\text{J}}}R_{\DD'_{\text{J}}}(\SS )$.

i) We prove $\max_{\MM \subseteq \NN} \min_{\SS \subseteq \MM}R_{\MM}(\SS )= \min_{\SS \subseteq \DD_{\text{J}}}R_{\DD_{\text{J}}}(\SS )$ by proving that: 1)  For any $\MM \cap \DD_{\text{J}} = \DD_{\text{J}}$, $\MM \neq \DD_{\text{J}}$, $\min_{\SS \subseteq \MM}R_{\MM}(\SS) \leq \min_{\SS \subseteq \DD_{\text{J}}}R_{\DD_{\text{J}}}(\SS)$. 2) For any $\MM \cap \DD_{\text{J}} \neq \DD_{\text{J}}$,  $\min_{\SS \subseteq \MM}R_{\MM}(\SS) \leq \min_{\SS \subseteq \MM \cup \DD_{\text{J}}}R_{\MM \cup \DD_{\text{J}}}(\SS)$, and thus $\min_{\SS \subseteq \MM}R_{\MM}(\SS)\leq \min_{\SS \subseteq \DD_{\text{J}}}R_{\DD_{\text{J}}}(\SS)$ by 1).
The details are as follows.

 1) Assume  $\MM \cap \DD_{\text{J}} = \DD_{\text{J}}$, $\MM \neq \DD_{\text{J}}$. We show $\min_{\SS \subseteq \MM}R_{\MM}(\SS) \leq \min_{\SS \subseteq \DD_{\text{J}}}R_{\DD_{\text{J}}}(\SS)$ by showing that for any $\SS \subseteq \DD_{\text{J}}$, $R_{\MM}(\SS \cup (\MM \setminus \DD_{\text{J}}  ) )\leq R_{ \DD_{\text{J}}}(\SS)$.

For any $\SS \subseteq \DD_{\text{J}}$,  by Part 2) of Lemma \ref{L:Runion}, we have $$R_{\MM}(\SS \cup (\MM \setminus \DD_{\text{J}}  ) )=R_{\DD_{\text{J}}  \cup (\MM \setminus \DD_{\text{J}})}(\SS \cup (\MM \setminus \DD_{\text{J}}  ) )= R_{ \DD_{\text{J}}}(\SS)+K_{ \DD_{\text{J}},\MM \setminus \DD_{\text{J}}}(\MM \setminus \DD_{\text{J}}  ).$$  We argue $K_{ \DD_{\text{J}},\MM \setminus \DD_{\text{J}}}(\MM \setminus \DD_{\text{J}}  )\leq 0$ by contradiction. Suppose $K_{ \DD_{\text{J}},\MM \setminus \DD_{\text{J}}}(\MM \setminus \DD_{\text{J}}  )>0$. Then, by Lemma \ref{exsit}, we have that there exists some nonempty $\mathcal{C} \subseteq \MM \setminus \DD_{\text{J}}  $ such that $K_{ \DD_{\text{J}},\CC}(\SS)> 0$, for any nonempty $ \SS \subseteq \mathcal{C}.$ This will further imply,  by Part 2) of Lemma \ref{union}, that $K_{ \DD_{\text{J}} \cup \CC}(\SS)> 0$, for any nonempty $ \SS \subseteq  \DD_{\text{J}} \cup \CC,$ which is in contradiction with the definition of $\DD_{\text{J}}$. Thus, we must have $K_{ \DD_{\text{J}},\MM \setminus \DD_{\text{J}}}(\MM \setminus \DD_{\text{J}}  )\leq 0$, and $R_{\MM}(\SS \cup (\MM \setminus \DD_{\text{J}}  ) )\leq R_{ \DD_{\text{J}}}(\SS)$.

2) Assume $\MM \cap \DD_{\text{J}} \neq \DD_{\text{J}}$. For any $\SS \subseteq \MM \cup \DD_{\text{J}}$, let $\SS_1=\mathcal{S}\MM$ and $\SS_2=\SS(\DD_{\text{J}} \setminus \MM)$. By Part 1) of Lemma \ref{L:Runion}, we have
$$R_{\MM \cup \DD_{\text{J}}}(\SS)=R_{\MM \cup (\DD_{\text{J}}\setminus \MM)}(\SS)\geq R_{\MM}(\SS_1)+K_{\MM \cup \DD_{\text{J}}}(\SS_2),$$
and then,
\begin{align*}
\min_{\SS \subseteq \MM \cup \DD_{\text{J}}}R_{\MM \cup \DD_{\text{J}}}(\SS)\geq &\min_{\SS \subseteq \MM \cup \DD_{\text{J}}}\{R_{\MM}(\SS_1)+K_{\MM \cup \DD_{\text{J}} }(\SS_2)\}\\
\geq &\min_{\SS \subseteq \MM \cup \DD_{\text{J}}}\{R_{\MM}(\SS_1)+K_{\DD_{\text{J}} }(\SS_2)\}\\
= &\min_{\SS_1 \subseteq \MM, \SS_2 \subseteq \DD_{\text{J}}\setminus \MM} \{R_{\MM}(\SS_1)+K_{\DD_{\text{J}} }(\SS_2)\}\\
=& \min_{\SS_1 \subseteq \MM} R_{\MM}(\SS_1)  + \min_{ \SS_2 \subseteq \DD_{\text{J}}\setminus \MM}K_{\DD_{\text{J}} }(\SS_2)\\
\geq &\min_{\SS_1 \subseteq \MM} R_{\MM}(\SS_1) ,
\end{align*}
where the last inequality follows from the fact that $K_{\DD_{\text{J}} }(\SS_2) >0$, for any nonempty $\SS_2 \subseteq \DD_{\text{J}}$.

ii) We can prove $\min_{\SS \subseteq \MM}R_{\MM}(\SS )< \min_{\SS \subseteq \DD'_{\text{J}}}R_{\DD'_{\text{J}}}(\SS )$, for any $\MM \nsubseteq \DD'_{\text{J}}$ by two similar steps as follows.

1) Through the similar lines as in Step 1) of Part i), we can prove $ \min_{\SS \subseteq \MM}R_{\MM}(\SS )< \min_{\SS \subseteq \DD'_{\text{J}}}R_{\DD'_{\text{J}}}(\SS )$, for any $\MM \cap \DD'_{\text{J}} = \DD'_{\text{J}}$, $\MM \neq \DD'_{\text{J}}$. The only difference is that here the inequality is strict, but it can be easily justified by noting that ``$=$'' is included in the definition of $\DD'_{\text{J}}$.

2) From Step 2) of Part i), it can be similarly proved that for any $\MM \cap \DD'_{\text{J}} \neq \DD'_{\text{J}}$, $\min_{\SS \subseteq \MM}R_{\MM}(\SS) \leq \min_{\SS \subseteq \MM \cup \DD'_{\text{J}}}R_{\MM \cup \DD'_{\text{J}}}(\SS)$. Therefore, if, further, $\MM \nsubseteq \DD'_{\text{J}}$,  then by 1) we have $$\min_{\SS \subseteq \MM}R_{\MM}(\SS)\leq \min_{\SS \subseteq \MM \cup \DD'_{\text{J}}}R_{\MM \cup \DD'_{\text{J}}}(\SS) < \min_{\SS \subseteq \DD'_{\text{J}}}R_{\DD'_{\text{J}}}(\SS).$$

iii) From Part ii), we have 1) $ \min_{\SS \subseteq \MM}R_{\MM}(\SS )< \min_{\SS \subseteq \DD'_{\text{J}}}R_{\DD'_{\text{J}}}(\SS )$, for any $\MM \cap \DD'_{\text{J}} = \DD'_{\text{J}}$, $\MM \neq \DD'_{\text{J}}$, and 2) for any $\MM \cap \DD'_{\text{J}} \neq \DD'_{\text{J}}$, $\min_{\SS \subseteq \MM}R_{\MM}(\SS) \leq \min_{\SS \subseteq \MM \cup \DD'_{\text{J}}}R_{\MM \cup \DD'_{\text{J}}}(\SS)\leq  \min_{\SS \subseteq \DD'_{\text{J}}}R_{\DD'_{\text{J}}}(\SS )$. Thus, it follows immediately that $ \min_{\SS \subseteq \DD'_{\text{J}}}R_{\DD'_{\text{J}}}(\SS )=\max_{\MM \subseteq \NN} \min_{\SS \subseteq \MM}R_{\MM}(\SS )$. This completes the proof of  Theorem \ref{T:equi}.
\end{proof}

\section{Conclusion} \label{conclusion}

Joint compression-message decoding introduced more freedom in selecting the compressions at the relays. Motivated by it, we have investigated the problem of finding the optimal compressions in maximizing the achievable rate of the original message. We have studied several different compress-and-forward relay schemes, and the unanimous conclusion is that the optimal compressions should always support successive compression-message decoding. In situations where compressions not supporting successive decoding have to be used, we have found that only those that can be jointly decoded are helpful to the decoding of the original message.

We have also developed a backward block-by-block decoding scheme. Compared to the repetitive encoding/all blocks united decoding scheme recently proposed in \cite{KimElGamal}, which improved the achievable rate in the multiple-relay case, we have realized that the key to the improvement comes from delaying the decoding until all the blocks have been finished. In retrospect, the multiple-relay case is different from the single-relay case in that it may take multiple blocks for the relays to help each other before their compressions can finally reach the destination. Hence, the block-by-block forward decoding scheme, which is sufficient for the single-relay case, may not work satisfactorily for multiple relays in general.

Finally, we need to point out that our discussion of optimality is restricted to the few selected compress-and-forward relay schemes. In generalizing the classical compress-and-forward relay scheme in \cite{covelg79} to the case of multiple relays, there could be many other choices of coding considerations \cite{kragasgup05}. Even for the single-relay case, the optimality of the original compression method used in \cite{covelg79} remains an open question (\cite{KimAllerton, wuxie09}).

\appendices

\section{Proofs of Lemmas \ref{L:union}-\ref{L:aub}}\label{A:A}
\begin{proof}[Proof of Lemma \ref{L:union}]

For any $\SS \subseteq \mathcal{A}\cup \mathcal{B}$, let $\SS_1=\mathcal{S}\mathcal{A}$ and $\SS_2=\SS(\BB \setminus \mathcal{A})$. Then,
\begin{align}
I_{\mathcal{A}\cup \mathcal{B}}(\SS)=&\sum_{i\in \SS} R_i - I(Y_\SS ;\hat Y_\SS|\hat Y_{(\mathcal{A}\cup \mathcal{B})\setminus \SS},Y)\nonumber \\
=&\sum_{i\in \SS_1} R_i - I(Y_{\SS_1} ;\hat Y_{\SS_1}|\hat Y_{(\mathcal{A}\cup \mathcal{B})\setminus \SS},Y)
+\sum_{i\in \SS_2} R_i - I(Y_{\SS_2} ;\hat Y_{\SS_2}|\hat Y_{(\mathcal{A}\cup \mathcal{B})\setminus \SS},\hat Y_{\SS_1},Y)\nonumber \\
\geq & \sum_{i\in \SS_1} R_i - I(Y_{\SS_1} ;\hat Y_{\SS_1}|\hat Y_{ \mathcal{A} \setminus \SS_1},Y)
+\sum_{i\in \SS_2} R_i - I(Y_{\SS_2} ;\hat Y_{\SS_2}|\hat Y_{ \mathcal{A}}, \hat Y_{ \mathcal{B} \setminus \SS_2},Y)\nonumber \\
=&I_{\mathcal{A}}(\SS_1)+I_{\mathcal{A},\mathcal{B}}(\SS_2)\label{E:fol1}\\
\geq &I_{\mathcal{A}}(\SS_1)+I_{\mathcal{B}}(\SS_2). \label{E:fol2}
\end{align}

If $I_{\mathcal{A}}(\SS_1)\geq 0$, $\forall \SS_1 \subseteq \mathcal{A}$, and $I_{\mathcal{B}}(\SS_2)\geq 0$, $\forall \SS_2 \subseteq \mathcal{B}$, then following \dref{E:fol2}, $I_{\mathcal{A}\cup \BB}(\SS)\geq 0$, $\forall \SS \subseteq \mathcal{A}\cup \BB.$
If $I_{\mathcal{A}}(\SS_1)\geq 0$, $\forall \SS_1 \subseteq \mathcal{A}$, and $I_{\mathcal{A},\mathcal{B}}(\SS_2)\geq 0$, $\forall \SS_2 \subseteq \mathcal{B}$, then following \dref{E:fol1}, $I_{\mathcal{A}\cup \BB}(\SS)\geq 0$, $\forall \SS \subseteq \mathcal{A}\cup \BB.$
\end{proof}

\begin{proof}[Proof of Lemma \ref{L:largestD}]
Let $\LL:=\{\mathcal{F} \subseteq \NN: I_{\mathcal{F}}(\SS)\geq 0, \forall \SS \subseteq \mathcal{F} \}$ and $\LL_{\max}:=\{\DD \in \LL: |\DD|=\max_{\mathcal{F} \in \LL}|\mathcal{F}| \}$. Suppose there are more than one element in $\LL_{\max}$, say, $\DD_1,\DD_2,\ldots,\DD_n$, where $n\geq 2$. Then based on 1) of Lemma \ref{L:union}, $\DD:=\bigcup_{i=1}^{n}\DD_i$ also satisfies that
$I_{\DD}(\SS)\geq 0, \forall \SS \subseteq \DD$, which is in contradiction, and hence Lemma \ref{L:largestD} is proved.
\end{proof}

\begin{proof}[Proof of Lemma \ref{L:exsit}]
If $I_{\mathcal{A},\mathcal{B}}(\SS)\geq 0$, $\forall \SS \subseteq \BB$, then this lemma obviously holds. Otherwise, if there exists some $\mathcal{S}_1 \subseteq \BB$, $\mathcal{S}_1 \neq \BB$, such that
$I_{\mathcal{A},\mathcal{B}}(\SS_1)< 0$, then we have $I_{\mathcal{A},\mathcal{B}}(\BB)- I_{\mathcal{A},\mathcal{B}}(\SS_1)\geq 0$, i.e.,
\begin{align*}
&\sum_{i\in \BB}R_i- I(Y_{\BB};\hat Y_{\BB}|\hat Y_{\mathcal{A}},Y)-\left(\sum_{i\in \SS_1}R_i- I(Y_{\SS_1};\hat Y_{\SS_1}|\hat Y_{\mathcal{A}},\hat Y_{\BB \setminus \SS_1}, Y)\right)\\
=&\sum_{i\in \BB \setminus \SS_1}R_i- I(Y_{\BB \setminus \SS_1};\hat Y_{\BB \setminus \SS_1}|\hat Y_{\mathcal{A}},Y) \\
=&I_{\mathcal{A},\mathcal{B} \setminus \mathcal{S}_1}(\mathcal{B} \setminus \mathcal{S}_1)\\
\geq & 0.
\end{align*}
Now, we arrive at the same situation as in the original assumption with $\mathcal{B}$ replaced by $\mathcal{B} \setminus \mathcal{S}_1$. Continue applying this argument, and we must be able to reach a nonempty $\CC \subseteq \BB$, such that $I_{\mathcal{A},\mathcal{C}}(\SS)\geq 0$, $\forall \SS \subseteq \CC$.
\end{proof}

\begin{proof}[Proof of Lemma \ref{L:aub}]
For any disjoint $\mathcal{A}$ and $\mathcal{B}$,
\begin{align*}
&I(\mathcal{A}\cup \mathcal{B})\\
=&\sum_{i \in \mathcal{A}\cup \mathcal{B}}R_i - I(Y_{\mathcal{A}\cup \mathcal{B}};\hat Y_{\mathcal{A}\cup \mathcal{B}}|\hat Y_{(\mathcal{A}\cup \mathcal{B})^c},Y) \\
=&\sum_{i \in \mathcal{A}}R_i - I(Y_{\mathcal{A}\cup \mathcal{B}};\hat Y_{\mathcal{A}}|\hat Y_{(\mathcal{A}\cup \mathcal{B})^c},Y)
+\sum_{i \in \mathcal{B}}R_i - I(Y_{\mathcal{A}\cup \mathcal{B}};\hat Y_{\mathcal{B}}|\hat Y_{(\mathcal{A}\cup \mathcal{B})^c},\hat Y_{\mathcal{A}}, Y)\\
=&\sum_{i \in \mathcal{A}}R_i - I(Y_{\mathcal{A}},\hat Y_{\mathcal{B}};\hat Y_{\mathcal{A}}|\hat Y_{(\mathcal{A}\cup \mathcal{B})^c},Y)
+\sum_{i \in \mathcal{B}}R_i - I(Y_{\mathcal{B}};\hat Y_{\mathcal{B}}|\hat Y_{(\mathcal{A}\cup \mathcal{B})^c},\hat Y_{\mathcal{A}}, Y)\\
=&\sum_{i \in \mathcal{A}}R_i - I(\hat Y_{\mathcal{B}};\hat Y_{\mathcal{A}}|\hat Y_{(\mathcal{A}\cup \mathcal{B})^c},Y)-I( Y_{\mathcal{A}};\hat Y_{\mathcal{A}}|\hat Y_{\mathcal{A}^c},Y)
+\sum_{i \in \mathcal{B}}R_i - I(Y_{\mathcal{B}};\hat Y_{\mathcal{B}}|\hat Y_{ \mathcal{B} ^c},  Y)\\
=&I(\mathcal{A})+I(\mathcal{B})- I(\hat Y_{\mathcal{A}};\hat Y_{\mathcal{B}}|\hat Y_{(\mathcal{A}\cup \mathcal{B})^c},Y),
\end{align*}
which proves the lemma.
\end{proof}

\section{Proofs of Lemmas  \ref{L:irrunion}-\ref{L:irraub}}\label{A:B}
\begin{proof}[Proof of Lemma \ref{L:irrunion}]

For any $\SS \subseteq \mathcal{A}\cup \mathcal{B}$, let $\SS_1=\mathcal{S}\mathcal{A}$ and $\SS_2=\SS(\BB \setminus \mathcal{A})$. Then,
\begin{align}
J_{\mathcal{A}\cup \mathcal{B}}(\SS)=&I(X_\SS; \hat Y_{(\mathcal{A}\cup \mathcal{B}) \setminus \SS},Y|X_{(\mathcal{A}\cup \mathcal{B}) \setminus \SS}) -I(Y_\SS; \hat Y_\SS|X_{\mathcal{A}\cup \mathcal{B}}, \hat Y_{(\mathcal{A}\cup \mathcal{B}) \setminus \SS},Y)\nonumber \\
=& I(X_{\SS_1}; \hat Y_{(\mathcal{A}\cup \mathcal{B}) \setminus \SS},Y|X_{(\mathcal{A}\cup \mathcal{B}) \setminus \SS})+I(X_{\SS_2}; \hat Y_{(\mathcal{A}\cup \mathcal{B}) \setminus \SS},Y|X_{\SS_1},X_{(\mathcal{A}\cup \mathcal{B}) \setminus \SS})                   \nonumber \\
&-I(Y_{\SS_1}; \hat Y_{\SS_1}|X_{\mathcal{A}\cup \mathcal{B}}, \hat Y_{(\mathcal{A}\cup \mathcal{B}) \setminus \SS},Y)
-I(Y_{\SS_2}; \hat Y_{\SS_2}|X_{\mathcal{A}\cup \mathcal{B}}, \hat Y_{\SS_1},\hat Y_{(\mathcal{A}\cup \mathcal{B}) \setminus \SS},Y)    \nonumber \\
=& I(X_{\SS_1}; \hat Y_{(\mathcal{A}\cup \mathcal{B}) \setminus \SS},Y|X_{(\mathcal{A}\cup \mathcal{B}) \setminus \SS})+I(X_{\SS_2}; \hat Y_{(\mathcal{A}\cup \mathcal{B}) \setminus \SS},Y|X_{\SS_1},X_{(\mathcal{A}\cup \mathcal{B}) \setminus \SS})                   \nonumber \\
&-[I(Y_{\SS_1}; \hat Y_{\SS_1}|X_{\mathcal{A}}, \hat Y_{\mathcal{A}\setminus \SS_1},Y)
-I(\hat Y_{\SS_1};X_{\mathcal{B}\setminus \mathcal{A}},\hat Y_{\mathcal{B}\mathcal{A}^c\setminus \SS_2}|X_{\mathcal{A}},\hat Y_{\mathcal{A}\setminus \SS_1},Y) ]\nonumber \\
&-I(Y_{\SS_2}; \hat Y_{\SS_2}|X_{\mathcal{A}\cup \mathcal{B}}, \hat Y_{\SS_1},\hat Y_{(\mathcal{A}\cup \mathcal{B}) \setminus \SS},Y)    \nonumber \\
=&[I(X_{\SS_1}; \hat Y_{(\mathcal{A}\cup \mathcal{B}) \setminus \SS},Y|X_{(\mathcal{A}\cup \mathcal{B}) \setminus \SS})- I(Y_{\SS_1}; \hat Y_{\SS_1}|X_{\mathcal{A}}, \hat Y_{\mathcal{A}\setminus \SS_1},Y)]\nonumber \\
&+[I(X_{\SS_2}; \hat Y_{(\mathcal{A}\cup \mathcal{B}) \setminus \SS},Y|X_{\SS_1},X_{(\mathcal{A}\cup \mathcal{B}) \setminus \SS})+I(\hat Y_{\SS_1};X_{\mathcal{B}\setminus \mathcal{A}},\hat Y_{\mathcal{B}\mathcal{A}^c\setminus \SS_2}|X_{\mathcal{A}},\hat Y_{\mathcal{A}\setminus \SS_1},Y)]\nonumber \\
&-I(Y_{\SS_2}; \hat Y_{\SS_2}|X_{\mathcal{A}}, X_{\mathcal{B}},\hat Y_{\mathcal{A}},\hat Y_{\mathcal{B} \setminus \SS_2},Y)    \nonumber \\
\geq&[I(X_{\SS_1}; \hat Y_{\mathcal{A} \setminus \SS_1},Y|X_{\mathcal{A} \setminus \SS_1})- I(Y_{\SS_1}; \hat Y_{\SS_1}|X_{\mathcal{A}}, \hat Y_{\mathcal{A}\setminus \SS_1},Y)]\nonumber \\
&+[I(X_{\SS_2}; \hat Y_{(\mathcal{A}\cup \mathcal{B}) \setminus \SS},Y|X_{\SS_1},X_{(\mathcal{A}\cup \mathcal{B}) \setminus \SS})+I(\hat Y_{\SS_1};X_{\mathcal{B}\setminus \mathcal{A}},\hat Y_{\mathcal{B}\mathcal{A}^c\setminus \SS_2}|X_{\mathcal{A}},\hat Y_{\mathcal{A}\setminus \SS_1},Y)]\nonumber \\
&-I(Y_{\SS_2}; \hat Y_{\SS_2}|X_{\mathcal{A}}, X_{\mathcal{B}},\hat Y_{\mathcal{A}},\hat Y_{\mathcal{B} \setminus \SS_2},Y)    \nonumber \\
=&[I(X_{\SS_2}; \hat Y_{(\mathcal{A}\cup \mathcal{B}) \setminus \SS},Y|X_{\SS_1},X_{(\mathcal{A}\cup \mathcal{B}) \setminus \SS})+I(\hat Y_{\SS_1};X_{\mathcal{S}_2},X_{\mathcal{B}\mathcal{A}^c \setminus \SS_2}, \hat Y_{\mathcal{B}\mathcal{A}^c\setminus \SS_2}|X_{\mathcal{A}},\hat Y_{\mathcal{A}\setminus \SS_1},Y)]\nonumber \\
&-I(Y_{\SS_2}; \hat Y_{\SS_2}|X_{\mathcal{A}}, X_{\mathcal{B}},\hat Y_{\mathcal{A}},\hat Y_{\mathcal{B} \setminus \SS_2},Y)+J_{\mathcal{A}}(\SS_1)  \nonumber \\
\geq &[I(X_{\SS_2}; \hat Y_{(\mathcal{A}\cup \mathcal{B}) \setminus \SS},Y|X_{\mathcal{A}},X_{\mathcal{B} \setminus \SS_2})+I(\hat Y_{\SS_1};X_{\mathcal{S}_2}|X_{\mathcal{A}},X_{\mathcal{B}\mathcal{A}^c \setminus \SS_2}, \hat Y_{\mathcal{B}\mathcal{A}^c\setminus \SS_2},\hat Y_{\mathcal{A}\setminus \SS_1},Y)]\nonumber \\
&-I(Y_{\SS_2}; \hat Y_{\SS_2}|X_{\mathcal{A}}, X_{\mathcal{B}},\hat Y_{\mathcal{A}},\hat Y_{\mathcal{B} \setminus \SS_2},Y)+J_{\mathcal{A}}(\SS_1)  \nonumber \\
=&I(X_{\SS_2}; \hat Y_{\mathcal{A}},\hat Y_{\mathcal{B} \setminus \SS_2},Y|X_{\mathcal{A}},X_{\mathcal{B} \setminus \SS_2})
-I(Y_{\SS_2}; \hat Y_{\SS_2}|X_{\mathcal{A}}, X_{\mathcal{B}},\hat Y_{\mathcal{A}},\hat Y_{\mathcal{B} \setminus \SS_2},Y)+J_{\mathcal{A}}(\SS_1)  \nonumber \\
=&J_{\mathcal{A}}(\SS_1)+J_{\mathcal{A},\mathcal{B}}(\SS_2)\label{E:IRRfol1}\\
\geq &J_{\mathcal{A}}(\SS_1)+J_{\mathcal{B}}(\SS_2). \label{E:IRRfol2}
\end{align}

If $J_{\mathcal{A}}(\SS_1)\geq 0$, $\forall \SS_1 \subseteq \mathcal{A}$, and $J_{\mathcal{B}}(\SS_2)\geq 0$, $\forall \SS_2 \subseteq \mathcal{B}$, then following \dref{E:IRRfol2}, $J_{\mathcal{A}\cup \BB}(\SS)\geq 0$, $\forall \SS \subseteq \mathcal{A}\cup \BB.$
If $J_{\mathcal{A}}(\SS_1)\geq 0$, $\forall \SS_1 \subseteq \mathcal{A}$, and $J_{\mathcal{A},\mathcal{B}}(\SS_2)\geq 0$, $\forall \SS_2 \subseteq \mathcal{B}$, then following \dref{E:IRRfol1}, $J_{\mathcal{A}\cup \BB}(\SS)\geq 0$, $\forall \SS \subseteq \mathcal{A}\cup \BB.$
\end{proof}

\begin{proof}[Proof of Lemma \ref{L:irrlargestD}]
Let $\LL:=\{\mathcal{F} \subseteq \NN: J_{\mathcal{F}}(\SS)\geq 0, \forall \SS \subseteq \mathcal{F} \}$ and $\LL_{\max}:=\{\DD \in \LL: |\DD|=\max_{\mathcal{F} \in \LL}|\mathcal{F}| \}$. Suppose there are more than one elements in $\LL_{\max}$, say, $\DD_1,\DD_2,\ldots,\DD_n$, where $n\geq 2$. Then based on 1) of Lemma \ref{L:irrunion}, $\DD:=\bigcup_{i=1}^{n}\DD_i$ also satisfies that
$J_{\DD}(\SS)\geq 0, \forall \SS \subseteq \DD$, which is in contradiction, and hence Lemma \ref{L:irrlargestD} is proved.
\end{proof}

\begin{proof}[Proof of Lemma \ref{L:irrexsit}]
If $J_{\mathcal{A},\mathcal{B}}(\SS)\geq 0$, $\forall \SS \subseteq \BB$, then this lemma obviously holds. Otherwise, if there exists some $\mathcal{S}_1 \subseteq \BB$, $\mathcal{S}_1 \neq \BB$, such that
$J_{\mathcal{A},\mathcal{B}}(\SS_1)< 0$, then we have $J_{\mathcal{A},\mathcal{B}}(\BB)- J_{\mathcal{A},\mathcal{B}}(\SS_1)\geq 0$, i.e.,
\begin{align*}
&I(X_{\mathcal{B}}; \hat Y_{\mathcal{A}}, Y|X_{\mathcal{A}}) -  I(Y_{\mathcal{B}};\hat Y_{\mathcal{B}}|X_{\mathcal{A}},\hat Y_{\mathcal{A}},Y, X_{\BB})\\
&-I(X_{\SS_1}; \hat Y_{\BB \setminus \SS_1},\hat Y_{\mathcal{A}}, Y|X_{\mathcal{A}},X_{\BB \setminus \SS_1}) + I(Y_{\SS_1};\hat Y_{\SS_1}|X_{\mathcal{A}},\hat Y_{\mathcal{A}},Y, X_{\BB}, \hat Y_{\BB\setminus \SS_1})\\
=&I(X_{\mathcal{B}\setminus \SS_1}; \hat Y_{\mathcal{A}}, Y|X_{\mathcal{A}})+I(X_{\SS_1}; \hat Y_{\mathcal{A}}, Y|X_{\mathcal{A}},X_{\mathcal{B}\setminus \SS_1})\\
&-I(Y_{\mathcal{B}\setminus \SS_1};\hat Y_{\mathcal{B}\setminus \SS_1}|X_{\mathcal{A}},\hat Y_{\mathcal{A}},Y, X_{\BB})
-I(Y_{\SS_1};\hat Y_{  \SS_1}|X_{\mathcal{A}},\hat Y_{\mathcal{A}},Y, X_{\BB},\hat Y_{\mathcal{B}\setminus \SS_1})\\
&-I(X_{\SS_1}; \hat Y_{\mathcal{A}}, Y|X_{\mathcal{A}},X_{\mathcal{B}\setminus \SS_1})-I(X_{\SS_1}; \hat Y_{\BB \setminus \SS_1}|\hat Y_{\mathcal{A}}, Y,X_{\mathcal{A}},X_{\BB \setminus \SS_1})+ I(Y_{\SS_1};\hat Y_{\SS_1}|X_{\mathcal{A}},\hat Y_{\mathcal{A}},Y, X_{\BB}, \hat Y_{\BB\setminus \SS_1})\\
=&I(X_{\mathcal{B}\setminus \SS_1}; \hat Y_{\mathcal{A}}, Y|X_{\mathcal{A}})-I(Y_{\BB \setminus \SS_1}; \hat Y_{\BB \setminus \SS_1}|\hat Y_{\mathcal{A}}, Y,X_{\mathcal{A}},X_{\BB \setminus \SS_1})\\
=&J_{\mathcal{A},\mathcal{B}\setminus \SS_1}(\mathcal{B}\setminus \SS_1)\\
\geq &0.
\end{align*}
Now, we arrive at the same situation as in the original assumption with $\mathcal{B}$ replaced by $\mathcal{B} \setminus \mathcal{S}_1$. Continue applying this argument, and we must be able to reach a nonempty $\CC \subseteq \BB$, such that $J_{\mathcal{A},\mathcal{C}}(\SS)\geq 0$, $\forall \SS \subseteq \CC$.
\end{proof}

\begin{proof}[Proof of Lemma \ref{L:irraub}]
For any disjoint $\mathcal{A}$ and $\mathcal{B}$,
\begin{align*}
&J(\mathcal{A}\circ \mathcal{B})\\
=&J(\mathcal{A})+J(\mathcal{B})-J(\mathcal{A} \cup \mathcal{B})\\
=&I(X_{\mathcal{A}};\hat Y_{\mathcal{A}^c},Y|X_{\mathcal{A}^c})-I( Y_{\mathcal{A}};\hat Y_{\mathcal{A}}|X_\NN,Y,\hat Y_{\mathcal{A}^c})\\
&+I(X_{\mathcal{B}};\hat Y_{\mathcal{B}^c},Y|X_{\mathcal{B}^c})-I( Y_{\mathcal{B}};\hat Y_{\mathcal{B}}|X_\NN,Y,\hat Y_{\mathcal{B}^c})\\
&-I(X_{ \mathcal{B}};\hat Y_{(\mathcal{A}\cup \mathcal{B})^c},Y|X_{(\mathcal{A}\cup \mathcal{B})^c})-I(X_{ \mathcal{A}};\hat Y_{(\mathcal{A}\cup \mathcal{B})^c},Y|X_{\mathcal{A}^c})\\
&+I( Y_{\mathcal{A}};\hat Y_{\mathcal{A}}|X_\NN,Y,\hat Y_{(\mathcal{A}\cup \mathcal{B})^c})+I( Y_{\mathcal{B}};\hat Y_{\mathcal{B}}|X_\NN,Y,\hat Y_{\mathcal{B}^c})\\
=&I(X_{\mathcal{A}};\hat Y_{\mathcal{B}}|X_{\mathcal{A}^c},\hat Y_{(\mathcal{A}\cup \mathcal{B})^c},Y)+I(X_{\mathcal{B}};X_{\mathcal{A}},\hat Y_{\mathcal{A}}|X_{(\mathcal{A}\cup \mathcal{B})^c},\hat Y_{(\mathcal{A}\cup \mathcal{B})^c},Y)
+I( \hat Y_{\mathcal{A}};\hat Y_{\mathcal{B}}|X_\NN,Y,\hat Y_{(\mathcal{A}\cup \mathcal{B})^c})\\
=&I(X_{\mathcal{A}},\hat Y_{\mathcal{A}};\hat Y_{\mathcal{B}}|X_{\mathcal{A}^c},\hat Y_{(\mathcal{A}\cup \mathcal{B})^c},Y)+I(X_{\mathcal{B}};X_{\mathcal{A}},\hat Y_{\mathcal{A}}|X_{(\mathcal{A}\cup \mathcal{B})^c},\hat Y_{(\mathcal{A}\cup \mathcal{B})^c},Y)\\
=&I(X_{\mathcal{B}},\hat Y_{\mathcal{B}};X_{\mathcal{A}},\hat Y_{\mathcal{A}}|X_{(\mathcal{A}\cup \mathcal{B})^c},\hat Y_{(\mathcal{A}\cup \mathcal{B})^c},Y),
\end{align*}
which proves the lemma.
\end{proof}

\section{Proof of Lemma \ref{L:Runion}}\label{A:C}
For any disjoint $\mathcal{A}$ and $\mathcal{B}$, and any $\SS \subseteq \mathcal{A}\cup \BB$, let $\SS_1=\mathcal{S}\mathcal{A}$ and $\SS_2=\SS \BB $. Then, we have
\begin{align}
R_{\mathcal{A}\cup \BB}(\SS)=&I(X,X_\SS; \hat Y_{(\mathcal{A}\cup \BB )\setminus \SS},Y|X_{(\mathcal{A}\cup \BB) \setminus \SS}) -I(Y_\SS; \hat Y_\SS|X,X_{\mathcal{A}\cup \BB}, \hat Y_{(\mathcal{A}\cup \BB) \setminus \SS},Y)\nonumber  \\
=&I(X,X_{\SS_1},X_{\SS_2}; \hat Y_{\mathcal{A} \setminus \SS_1},\hat Y_{\mathcal{B} \setminus \SS_2},Y|X_{\mathcal{A} \setminus \SS_1},X_{\mathcal{B} \setminus \SS_2}) -I(Y_{(\mathcal{S}_1\cup \SS_2 )}; \hat Y_{(\mathcal{S}_1\cup \SS_2 )}|X,X_{\mathcal{A}}, X_\BB, \hat Y_{\mathcal{A} \setminus \SS_1},\hat Y_{\mathcal{B} \setminus \SS_2},Y)\nonumber \\
=&I(X,X_{\SS_1}; \hat Y_{\mathcal{A} \setminus \SS_1},\hat Y_{\mathcal{B} \setminus \SS_2},Y|X_{\mathcal{A} \setminus \SS_1},X_{\mathcal{B} \setminus \SS_2})
+I(X_{\SS_2}; \hat Y_{\mathcal{A} \setminus \SS_1},\hat Y_{\mathcal{B} \setminus \SS_2},Y|X,X_{\mathcal{A} },X_{\mathcal{B} \setminus \SS_2})  \nonumber  \\
&-[I(Y_{\mathcal{S}_1}; \hat Y_{\mathcal{S}_1}|X,X_{\mathcal{A}}, X_\BB, \hat Y_{\mathcal{A} \setminus \SS_1},\hat Y_{\mathcal{B} \setminus \SS_2},Y)
+I(Y_{ \SS_2 }; \hat Y_{ \SS_2}|X,X_{\mathcal{A}}, X_\BB, \hat Y_{\mathcal{A}},\hat Y_{\mathcal{B} \setminus \SS_2},Y)]\nonumber  \\
=&I(X,X_{\SS_1}; \hat Y_{\mathcal{A} \setminus \SS_1},Y|X_{\mathcal{A} \setminus \SS_1})
+I(X,X_{\SS_1};X_{\mathcal{B} \setminus \SS_2}, \hat Y_{\mathcal{B} \setminus \SS_2}|X_{\mathcal{A} \setminus \SS_1},\hat Y_{\mathcal{A} \setminus \SS_1},Y)  \nonumber  \\
&-[I(Y_{\mathcal{S}_1}; \hat Y_{\mathcal{S}_1}|X,X_{\mathcal{A}}, \hat Y_{\mathcal{A} \setminus \SS_1},Y)
-I(X_\BB,\hat Y_{\mathcal{B} \setminus \SS_2}; \hat Y_{\mathcal{S}_1}|X,X_{\mathcal{A}},  \hat Y_{\mathcal{A} \setminus \SS_1},Y)]\nonumber  \\
&+I(X_{\SS_2}; \hat Y_{\mathcal{A} \setminus \SS_1},\hat Y_{\mathcal{B} \setminus \SS_2},Y|X,X_{\mathcal{A} },X_{\mathcal{B} \setminus \SS_2})
-I(Y_{ \SS_2 }; \hat Y_{ \SS_2}|X,X_{\mathcal{A}}, X_\BB, \hat Y_{\mathcal{A}},\hat Y_{\mathcal{B} \setminus \SS_2},Y)\nonumber  \\
=&[I(X,X_{\SS_1}; \hat Y_{\mathcal{A} \setminus \SS_1},Y|X_{\mathcal{A} \setminus \SS_1})-I(Y_{\mathcal{S}_1}; \hat Y_{\mathcal{S}_1}|X,X_{\mathcal{A}}, \hat Y_{\mathcal{A} \setminus \SS_1},Y)] \nonumber  \\
&+I(X,X_{\SS_1};X_{\mathcal{B} \setminus \SS_2}, \hat Y_{\mathcal{B} \setminus \SS_2}|X_{\mathcal{A} \setminus \SS_1},\hat Y_{\mathcal{A} \setminus \SS_1},Y)
+I(X_\BB,\hat Y_{\mathcal{B} \setminus \SS_2}; \hat Y_{\mathcal{S}_1}|X,X_{\mathcal{A}},  \hat Y_{\mathcal{A} \setminus \SS_1},Y)\nonumber  \\
&+I(X_{\SS_2}; \hat Y_{\mathcal{A} \setminus \SS_1},\hat Y_{\mathcal{B} \setminus \SS_2},Y|X,X_{\mathcal{A} },X_{\mathcal{B} \setminus \SS_2})
-I(Y_{ \SS_2 }; \hat Y_{ \SS_2}|X,X_{\mathcal{A}}, X_\BB, \hat Y_{\mathcal{A}},\hat Y_{\mathcal{B} \setminus \SS_2},Y)\nonumber \\
=&R_{\mathcal{A}}(\SS_1) +I(X,X_{\SS_1};X_{\mathcal{B} \setminus \SS_2}, \hat Y_{\mathcal{B} \setminus \SS_2}|X_{\mathcal{A} \setminus \SS_1},\hat Y_{\mathcal{A} \setminus \SS_1},Y)
+I(X_\BB,\hat Y_{\mathcal{B} \setminus \SS_2}; \hat Y_{\mathcal{S}_1}|X,X_{\mathcal{A}},  \hat Y_{\mathcal{A} \setminus \SS_1},Y)\nonumber  \\
&+I(X_{\SS_2}; \hat Y_{\mathcal{A} \setminus \SS_1},\hat Y_{\mathcal{B} \setminus \SS_2},Y|X,X_{\mathcal{A} },X_{\mathcal{B} \setminus \SS_2})
-I(Y_{ \SS_2 }; \hat Y_{ \SS_2}|X,X_{\mathcal{A}}, X_\BB, \hat Y_{\mathcal{A}},\hat Y_{\mathcal{B} \setminus \SS_2},Y)\label{E:lemmafollow}.
\end{align}

When $\SS_2=\BB$, following  \dref{E:lemmafollow}, we have
\begin{align*}
R_{\mathcal{A}\cup \BB}(\SS)=&R_{\mathcal{A}}(\SS_1)
+I(X_\BB; \hat Y_{\mathcal{S}_1}|X,X_{\mathcal{A}},  \hat Y_{\mathcal{A} \setminus \SS_1},Y) \\
&+I(X_{\BB}; \hat Y_{\mathcal{A} \setminus \SS_1}, Y|X,X_{\mathcal{A} })
-I(Y_{ \BB }; \hat Y_{ \BB}|X,X_{\mathcal{A}}, X_\BB, \hat Y_{\mathcal{A}},Y)\\
=&R_{\mathcal{A}}(\SS_1)
+I(X_{\BB}; \hat Y_{\mathcal{A} }, Y|X,X_{\mathcal{A} })
-I(Y_{ \BB }; \hat Y_{ \BB}|X,X_{\mathcal{A}}, X_\BB, \hat Y_{\mathcal{A}},Y)\\
=&R_{\mathcal{A} }(\SS_1)+K_{\mathcal{A},\BB  }(\BB ).
\end{align*}

 Generally, for any $\SS_2 \subseteq \BB$, continuing  \dref{E:lemmafollow}, we have
\begin{align*}
R_{\mathcal{A}\cup \BB}(\SS)
\geq &R_{\mathcal{A}}(\SS_1)
+I(X_\BB,\hat Y_{\mathcal{B} \setminus \SS_2}; \hat Y_{\mathcal{S}_1}|X,X_{\mathcal{A}},  \hat Y_{\mathcal{A} \setminus \SS_1},Y)\nonumber  \\
&+I(X_{\SS_2}; \hat Y_{\mathcal{A} \setminus \SS_1},\hat Y_{\mathcal{B} \setminus \SS_2},Y|X,X_{\mathcal{A} },X_{\mathcal{B} \setminus \SS_2})
-I(Y_{ \SS_2 }; \hat Y_{ \SS_2}|X,X_{\mathcal{A}}, X_\BB, \hat Y_{\mathcal{A}},\hat Y_{\mathcal{B} \setminus \SS_2},Y)\\
= &R_{\mathcal{A}}(\SS_1)
+I(X_{\mathcal{B} \setminus \SS_2},\hat Y_{\mathcal{B} \setminus \SS_2}; \hat Y_{\mathcal{S}_1}|X,X_{\mathcal{A}},  \hat Y_{\mathcal{A} \setminus \SS_1},Y)
+I(X_{ \SS_2}; \hat Y_{\mathcal{S}_1}|X,X_{\mathcal{A}}, X_{\mathcal{B} \setminus \SS_2},\hat Y_{\mathcal{B} \setminus \SS_2}, \hat Y_{\mathcal{A} \setminus \SS_1},Y)  \\
&+I(X_{\SS_2}; \hat Y_{\mathcal{A} \setminus \SS_1},\hat Y_{\mathcal{B} \setminus \SS_2},Y|X,X_{\mathcal{A} },X_{\mathcal{B} \setminus \SS_2})
-I(Y_{ \SS_2 }; \hat Y_{ \SS_2}|X,X_{\mathcal{A}}, X_\BB, \hat Y_{\mathcal{A}},\hat Y_{\mathcal{B} \setminus \SS_2},Y)\\
= &R_{\mathcal{A}}(\SS_1)
+I(X_{\mathcal{B} \setminus \SS_2},\hat Y_{\mathcal{B} \setminus \SS_2}; \hat Y_{\mathcal{S}_1}|X,X_{\mathcal{A}},  \hat Y_{\mathcal{A} \setminus \SS_1},Y)   \\
&+I(X_{\SS_2}; \hat Y_{\mathcal{A} },\hat Y_{\mathcal{B} \setminus \SS_2},Y|X,X_{\mathcal{A} },X_{\mathcal{B} \setminus \SS_2})
-I(Y_{ \SS_2 }; \hat Y_{ \SS_2}|X,X_{\mathcal{A}}, X_\BB, \hat Y_{\mathcal{A}},\hat Y_{\mathcal{B} \setminus \SS_2},Y)\\
\geq &R_{\mathcal{A}}(\SS_1)
+I(X_{\SS_2}; \hat Y_{\mathcal{A} },\hat Y_{\mathcal{B} \setminus \SS_2},Y|X,X_{\mathcal{A} },X_{\mathcal{B} \setminus \SS_2})
-I(Y_{ \SS_2 }; \hat Y_{ \SS_2}|X,X_{\mathcal{A}}, X_\BB, \hat Y_{\mathcal{A}},\hat Y_{\mathcal{B} \setminus \SS_2},Y)\\
=&R_{\mathcal{A}}(\SS_1) + K_{\mathcal{A}\cup \mathcal{B}}(\SS_2).
\end{align*}
This completes the proof of  Lemma \ref{L:Runion}.

\end{document}